\documentclass[11pt,a4paper]{article}
\usepackage{microtype}
\usepackage{amssymb, amsmath, amsthm}
\usepackage{fullpage}
\usepackage{enumerate}
\usepackage{graphicx}
\usepackage{subfig}
\usepackage{times}
\usepackage{pifont}
\usepackage{hyperref}
\usepackage{wrapfig}
\usepackage{times}

 \theoremstyle{plain}

\def\cNP{\hbox{\rm \sffamily NP}}

\newtheorem{theorem}{Theorem}[section]
\newtheorem{lemma}[theorem]{Lemma}

\newtheorem{observation}[theorem]{Observation}

\theoremstyle{definition}

\begin{document}

\title{Bounded embeddings of graphs in the plane}

\author{
{Radoslav Fulek\footnote{
The research leading to these results has received funding from the People Programme (Marie Curie Actions) of the European Union's Seventh Framework Programme (FP7/2007-2013) under REA grant agreement no [291734].}
}
}

\date{}

\maketitle

\begin{abstract}
A drawing in the plane ($\mathbb{R}^2$)  of a graph $G=(V,E)$ equipped with a function $\gamma: V \rightarrow \mathbb{N}$  is \emph{$x$-bounded} if (i) $x(u) <x(v)$ whenever $\gamma(u)<\gamma(v)$ 
and (ii) $\gamma(u)\leq\gamma(w)\leq \gamma(v)$, where $uv\in E$ and $\gamma(u)\leq \gamma(v)$, whenever $x(w)\in x(uv)$, where $x(.)$
denotes the projection to the $x$-axis.
We prove a characterization of isotopy classes of graph embeddings in the plane containing an $x$-bounded embedding.

  
Then we present an efficient algorithm, that relies on our result, for testing the existence of an $x$-bounded embedding if the given
 graph is a tree or generalized $\Theta$-graph.
 This partialy answers a question raised recently by Angelini et al. and Chang et al.,
 and proves that c-planarity testing of flat clustered graphs with
 three clusters is tractable if each connected component of the underlying abstract graph is a tree.
\end{abstract}

\newpage

\thispagestyle{empty}

\tableofcontents

\newpage
\setcounter{page}{1}

\section{Introduction}

Testing planarity of graphs with additional constraints is a popular theme in
the area of graph visualizations abundant with open problems mainly of
algorithmic nature.
Probably the most important open problem in the area
is to determine the complexity status, i.e., P, NP-hard, or IP, of the problem of deciding for a pair of (planar)  graphs $G_1$ and $G_2$, whose edge sets possibly intersect,
if there exists a drawing of $G_1\cup G_2$ in the plane,
whose restriction to both graphs, $G_1$ and $G_2$, is an embedding.
The problem, also known as SEFE-2, was introduced in 2003 by Brass et al. in~\cite{Brass2007117} and 
its prominence  was realized  by Schaefer in~\cite{S12+},
where polynomial time reductions of many problems in the area to SEFE-2 is given,
see Figure~2 therein.

Among the problems reducible to SEFE-2 in a polynomial time is a notoriously difficult open problem raised under the name of \emph{c-planarity} in 1995 by Feng, Cohen and Eades~\cite{Feng95,Feng95+}.
The problem asks for a given planar graph $G$ equipped with a hierarchical structure on its vertex 
set, i.e., clusters, to decide if a planar embedding $G$ with the following property exists:
the vertices in each cluster are drawn inside a disc so that the discs
form a laminar set family corresponding to the given hierarchical structure
and the embedding has the least possible number of edge-crossings with the boundaries of the discs.
Again we are interested in the complexity status of the problem.

On the other hand, quite well understood from the algorithmic perspective are 
upward embeddings of directed acyclic planar graphs~\cite{BBLM94,GT02} and closely related 
various layered drawings of leveled graphs~\cite{BBF04,JLM98}.
In the setting of layered drawings we  place the vertices on, e.g., parallel lines or concentric circles, corresponding to the levels of $G$. Furthermore, we  require that edges lie between the levels of their endpoints
and that edges are monotone in the sense that they intersect any line (circle) parallel to (concentric with) the chosen lines (circles) at most once.
Also these easier planarity variants are reducible in a polynomial time to SEFE-2~\cite{S12+}.
The layered drawings with  parallel lines representing levels
are called \emph{level drawings}.
The $x$-bounded planarity treated in this work sits complexity-wise 
between the level planarity and c-planarity.
Hence, a better understanding of $x$-bounded planarity is a vital step towards 
shifting the frontier between complexity-wise known and open planarity variants.

Let $(G,\gamma)$ denote a pair of a planar graph $G=(V,E)$ 
and a function  $\gamma: V \rightarrow \mathbb{N}$.
A drawing in the plane ($\mathbb{R}^2$)  of $G$ is \emph{$x$-bounded} if (i) $x(u) <x(v)$ whenever $\gamma(u)<\gamma(v)$ 
and (ii) $\gamma(u)\leq\gamma(w)\leq \gamma(v)$, where $uv\in E$ and $\gamma(u)\leq \gamma(v)$, whenever $x(w)\in x(uv)$, where $x(.)$
denotes the projection to the $x$-axis,
see Figure~\ref{fig:strips3} for an illustration.
As a consequence of the proof of Theorem~\ref{thm:linearly} (below)
there exists an $x$-bounded embedding of $(G,\gamma)$ in which
projection $x(e)$ of every edge  $e\in E$ is injective, i.e., \emph{$x$-monotone}, as  soon as there exists an arbitrary  $x$-bounded embedding of $(G,\gamma)$.

\begin{lemma}
\label{lemma:intro}
There exists an $x$-bounded embedding of $(G,\gamma)$  in which
the projection $x(e)$ of every edge  $e\in E$ is injective if there exists an arbitrary  $x$-bounded embedding of $(G,\gamma)$.
\end{lemma}

Hence, we will not lose generality if we are interested only
in finding an $x$-monotone embedding that is $x$-bounded.
For that reason we call an $x$-bounded drawing \emph{an $x$-bounded embedding}
if it is edge-crossing free and $x(e)$ is injective for every edge $e\in E$, see Figure~\ref{fig:strips4} for an illustration.
Moreover, by ~\cite[Theorem 2]{PT04_monotone}
we can assume that edges in such embedding
are straight-line segments.
The main contribution of our work is a characterization of
isotopy classes of embeddings of $G$ in the plane containing
an $x$-bounded embedding Theorem~\ref{thm:main}.

\bigskip
\begin{figure}[htp]
\centering
\subfloat[]{\includegraphics[scale=0.7]{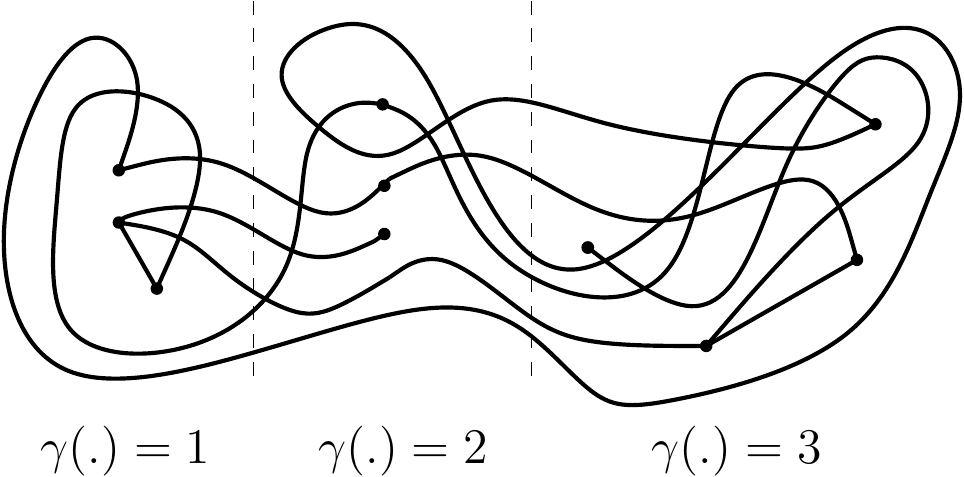}
    \label{fig:strips3}
	} \hspace{0.2cm}
\subfloat[]{\includegraphics[scale=0.7]{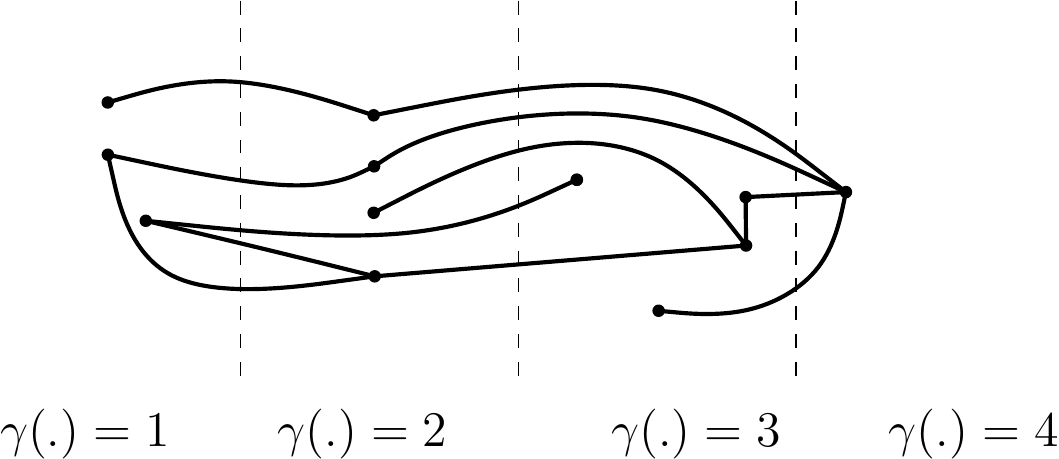}
	\label{fig:strips4}
	}
\caption{(a) An $x$-bounded drawing of a pair $(G,\gamma)$ each vertical strip contains vertices whose $\gamma$ value is the same; (b) An $x$-bounded embedding of a pair $(G,\gamma)$, $x$-monotone as required by our definition.}
\end{figure}

We use the characterization to prove the correctness of a PQ-tree based algorithm to test if an $x$-bounded embedding of $G$ exists. The characterization  turns the problem of the existence of an $x$-bounded embedding into a problem  that can be solved efficiently by employing  a PQ-tree based technique by Bl\"asius and Rutter~\cite{BR14+} at least in the case of trees and a union of internally disjoint paths between a pair of vertices. Moreover, we suspect that with additional twists the problem can be solved efficiently for any graph.
The characterization also implies a common generalization of the weak Hanani--Tutte theorem 
and its monotone variant by Pach and T\'oth, Theorem~\ref{thm:linearly}.

\subsection{Results}
Refer to Section~\ref{sec:notation} for the definitions.
Suppose that we have a pair of a graph and $\gamma$ as above, where $G$ is planar, connected, and let $\mathcal{E}$ denote the isotopy class
of an embedding of $G$ in the plane. Let us treat $\mathcal{E}$ as an embedded two-dimensional
polytopal complex, and let $\mathcal{C}=(\mathcal{E},\mathbb{Z}_2)$ be the corresponding
chain complex, i.e., in $\mathcal{C}$ two-dimensional chains are generated by the inner faces 
of $\mathcal{E}$, one-dimensional chains by the edges, etc.
The boundary operator $\partial(.)$ is defined 
as usual, i.e., we put $\partial(v)=\emptyset$, for any $v\in V$,
and hence, $\gamma(\partial(v))=\emptyset$.
Let $i_\mathcal{E}(C_1,C_2)$ denote the algebraic intersection number
of the supports of pure chains $C_1$ and $C_2$ in  $\mathcal{E}$  such that $dim(C_1)+dim(C_2)=2$,
where $dim(.)$ is dimension, and 
 the support of both $C_1$ and $C_2$ is homeomorphic to an orientable manifold of the corresponding dimension.
Our main result is the following.

\begin{theorem}
\label{thm:main}
The isotopy class $\mathcal{E}$ contains an $x$-bounded embedding  if and only if
  $i_\mathcal{E}(C_1,C_2)= 0$ whenever $\gamma(C_1) \cap \gamma(\partial C_2) = \emptyset$ and $\gamma(\partial C_1) \cap \gamma(C_2) = \emptyset$, where $\gamma(.)$ is extended over $\mathbb{R}$ linearly 
  to edges.\footnote{It is enough to
  consider  pairs $C_1$ and $C_2$, where both
  $C_1$ and $C_2$ are homeomorphic to a ball of the corresponding dimension.}
\end{theorem}  

We remark that ``only if'' part of the theorem
is easy, and thus, it is the ``if'' part that is
interesting. 
Instead of proving Theorem~\ref{thm:main}
we prove its equivalent reformulation,
Theorem~\ref{thm:characterization}, that is less conceptual, but more convenient to work with. The characterization was extracted from the proof of a weak variant of the Hanani--Tutte theorem~\cite{F14} in the setting of strip clustered graphs.
However, the proof of Theorem~\ref{thm:main} presented here is quite different,
and adapts ideas of Minc~\cite{M97} and M.~Skopenkov~\cite{S03}.

As an application of our characterization we generalize the aforementioned  variant of the Hanani--Tutte theorem.

\begin{theorem}\label{thm:linearly}
If $(G,\gamma)$ admits an $x$-bounded drawing $\mathcal{E}$ in which every pair of edges cross evenly
then $(G,\gamma)$ admits an $x$-bounded embedding.
Moreover, there exists an $x$-bounded embedding of $(G,\gamma)$ with the same rotation system as in $\mathcal{E}$.
\end{theorem}

The previous theorem is a special case of a corollary of a result Skopenkov~\cite[Theorem 1.5]{S03} and an extension of the following result of Pach and T\'{o}th.

\begin{theorem}\label{thm:mono}
Let $G$ denote a graph whose vertices are totally ordered.
Suppose that there exists a drawing $\mathcal{E}$ of $G$, in which $x$-coordinates of vertices respect their order, edges are $x$-monotone and every pair of edges cross an even number of times.
Then there exists an embedding of $G$, in which the vertices are drawn as in $\mathcal{E}$, the edges are $x$-monotone, and the rotation system is the same as in $\mathcal{E}$.
\end{theorem}

To support our conjecture we prove the strong variant of Theorem~\ref{thm:linearly} under
the condition that the underlying abstract graph $G$ is a subdivision of a vertex three-connected graph.
In general, we only know that this variant is true for two clusters~\cite{FKMP15}.

\begin{theorem}\label{thm:linearlyStrong}
Let $G$ denote a subdivision of a vertex three-connected graph.
If  $(G,\gamma)$ admits an independently even $x$-bounded drawing  then $(G,\gamma)$
admits an $x$-bounded embedding.\footnote{The argument in the proof of Theorem~\ref{thm:linearlyStrong}
proves, in fact, a strong variant even in the case, when we require the vertices participating in a cut or two-cut to have the maximum degree three. Hence, we obtained a polynomial time algorithm even in the case of sub-cubic cuts and two-cuts.}
\end{theorem}

The strong variant of Theorem~\ref{thm:linearly} (see Section~\ref{sec:hanani} for the explanation of what is meant by the ``strong and weak variant''), which is conjectured
to hold, would imply the existence
of a polynomial time algorithm for the corresponding variant of the
c-planarity testing~\cite{FKMP15}.
To the best of our knowledge, a polynomial time algorithm
was given only in the case, when the underlying planar graph has a prescribed isotopy class for the resulting embedding~\cite{ADDF13+}.
Our weak variant gives a polynomial time algorithm if $G$ is sub-cubic, and in the same case as~\cite{ADDF13+}.
Nevertheless, we think that the weak variant is interesting in its own right.

We give an algorithm for testing $x$-bounded embeddability for trees.
The algorithm works, in fact, with 0--1 matrices having some elements ambiguous, and
can be thought of as a special case of Simultaneous PQ-ordering considered recently by
Bl{\"{a}}sius and Rutter~\cite{BR14}. However, we not need any  result from~\cite{BR14}
in the case of trees.

\begin{theorem}\label{thm:AlgTreeXBounded}
We can test in cubic time if  $(G,\gamma)$ admits an $x$-bounded embedding
when the underlying abstract graph $G$ is a tree.
\end{theorem}

Using a more general variant of Simultaneous PQ-ordering we prove that $x$-bounded planarity is 
polynomial time solvable also when the abstract graph is a set of internally vertex disjoint paths joining a pair of vertices. We call such a graph a \emph{theta-graph}.
Unlike in the case of trees, in the case of theta-graphs we crucially rely on the main result
of~\cite{BR14}. The following theorem follows immediately from Theorem~\ref{thm:theta_alg}.

\begin{theorem}\label{thm:theta}
We can test in quartic time if  $(G,\gamma)$ admits an $x$-bounded embedding
when the underlying abstract graph $G$ is a theta-graph.
\end{theorem}

Similarly as for trees we are not aware of any previous algorithm with a polynomial running time in this case.

\section{Preliminaries}

\label{sec:notation}

\subsection{Notation}

\label{sec:notation}

%
%
%
%
%
%
%
%
%
%

\paragraph{Algebraic intersection number.}
Let $M_1$ and $M_2$, respectively, be $n_1$ and $n_2$-dimensional orientable manifold (possibly with boundaries) such that $n_1+n_2=n$.
Assume that $M_1$ and $M_2$ are PL embedded into $\mathbb{R}^n$ such that they are in general position,
i.e., they intersect in a finite set of points. Let us fix an orientation on $M_1$ and $M_2$.
The algebraic intersection number $i_\mathcal{E}(M_1,M_2) = \sum_{p}o(p)$, where we sum over all intersection
points $p$ of $M_1$ and $M_2$ and $o(p)$ is 1 is if the intersection point is positive and -1 if the 
intersection point is negative with respect to the chosen orientations.
If $M_1$ and $M_2$ are not in a general position $i_\mathcal{E}(M_1,M_2)$ denotes
$i_\mathcal{E}(M_1',M_2')$, where $M_1'$ and $M_2'$, respectively, is slightly perturbed $M_1$ and $M_2$.
(A perturbation eliminates ``touchings'' and does not introduce new ``crossings''.)
Note that $i_\mathcal{E}(M_1,M_2)=0$ is not affected by the choice of orientation.

 \paragraph{Graphs and its drawings.}
 Let $G=(V,E)$ denote a connected planar graph possibly with multi-edges but without loops.
A \emph{drawing} of $G$ is a representation of $G$ in the plane where every vertex
 in $V$ is represented by a unique point and every
edge $e=uv$ in $E$ is represented by a Jordan arc joining the two points that represent $u$ and $v$. 
We assume that in a drawing no edge passes through a vertex,
no two edges touch and every pair of edges cross in finitely many points.
An \emph{embedding} of $G$ is an  edge-crossing free drawing.
If it leads to no confusion, we do not distinguish between
a vertex or an edge and its representation in the drawing and we use the words ``vertex'' and ``edge'' in both
 contexts.
 Since in the problem we study connected components of $G$ can be treated separately, we can afford to assume that $G$ is connected throughout the paper.
A  \emph{face} in an embedding is a connected component of the complement of the embedding 
of $G$ (as a topological space) in the plane.
 The \emph{facial walk} of $f$ is the walk in $G$ with a fixed orientation that we obtain by traversing the boundary of $f$ counter-clockwise.
In order to simplify the notation we sometimes denote the facial walk of a face $f$ by $f$. 
The cardinality $|f|$ of $f$ denotes the number of edges (counted with multiplicities) in the facial walk of $f$.
Let $F$ denote a set of faces in an embedding. We let $G[F]$ denote the subgraph of $G$ induced by the edges incident to the faces of $F$.
  A pair of consecutive edges $e$ and $e'$ in a facial walk $f$ create a \emph{wedge} incident to $f$ at their common vertex.
  A vertex or an edge is \emph{incident} to a face $f$, if it appears on its facial walk.
The \emph{rotation} at a vertex is the counter-clockwise cyclic order of the end pieces of its incident edges
in a drawing of $G$.
 The \emph{rotation system} of a graph is the set of rotations at all its vertices.
An embedding of $G$ is up to an isotopy and the choice of an outer (unbounded) face described by the rotations at its vertices. We call such a description of an embedding of $G$ 
a \emph{combinatorial embedding}.
The \emph{interior} and \emph{exterior} of a cycle in an embedded graph is the bounded and unbounded, respectively, connected component
of its complement in the plane. 
Similarly, the \emph{interior} and \emph{exterior} of an inner face in an embedded graph is the bounded and unbounded, respectively, connected component
of the complement of its facial walk in the plane, and vice-versa for the outer face.
We when talking  about interior/exterior or area of a cycle  
in a graph $G$ with a combinatorial embedding and a \emph{designated} outer face  we mean it with respect to an embedding in the isotopy class that $G$ defines.
For $V'\subseteq V$ we denote by $G[V']$ the subgraph of $G$ induced by $V'$.

\paragraph{Simple and semi-simple faces.}
Let $\gamma:V \rightarrow \mathbb{N}$ be the given labeling of the vertices of $G$ by integers. Given a face $f$ in an embedding of $G$,
a vertex $v$ incident to $f$ is a \emph{local minimum} (\emph{maximum}) of $f$ if in the corresponding facial walk $W$ of $f$ the value
of $\gamma(v)$ is not bigger (not smaller) than the value of its successor and predecessor on $W$. A minimal and maximal, respectively, local
minimum and  maximum of $f$ is called  \emph{global minimum} and \emph{maximum} of $f$.
The face $f$ is \emph{simple} with respect to $\gamma$ if $f$ has exactly one local minimum and one local maximum.
The face $f$ is \emph{semi-simple} (with respect to $\gamma$) if $f$ has exactly two local minima and these minima have the same value, and two local maxima and these maxima have the same value.
A path $P$ is \emph{(strictly) monotone with respect to $\gamma$} if the labels of the  vertices on $P$ form a (strictly) monotone sequence if ordered
in the correspondence with their appearance on  $P$.

\paragraph{Clustering.}
Given a pair  $(G,\gamma)$ we naturally associate with it 
a partition of the vertex set into the cluster $V_i$'s such that $v$ belongs to $V_{\gamma(v)}$. We refer to the cluster whose vertices get label $i$ as to the $i^{\mathrm{th}}$ cluster. 
Let $(\overrightarrow{G},\gamma)$ denote the directedgraph obtained from $(G,\gamma)$ by orienting every edge
$uv$ from the vertex with the smaller label to the vertex with the bigger label, and in case of a tie orienting $uv$ arbitrarily.
A \emph{sink} and \emph{source}, respectively, of $\overrightarrow{G}$ is
a vertex with no outgoing and incoming edges.

\paragraph{Flat clustered graph.}
A \emph{flat clustered graph}, shortly  \emph{c-graph}, is a pair $(G,T)$, where $G=(V,E)$ is a graph and $T=\{V_0, \ldots, V_{c-1}\}$, $\biguplus_i V_i=V$, is a partition of the
vertex set into \emph{clusters}. 
A  c-graph $(G,T)$ is \emph{clustered planar} (or briefly \emph{c-planar}) if $G$ has an
 embedding in the plane such that (i)
for every $V_i\in T$ there is a topological disc $D(V_i)$, where $\mathrm{interior}(D(V_i))\cap \mathrm{interior} (D(V_j))=\emptyset$, if $i\not=j$,
 containing all the vertices of $V_i$ in its interior, and (ii)
 every edge of $G$ intersects the boundary of $D(V_i)$ at most once for every $D(V_i)$.
A c-graph  $(G,T)$ with a given combinatorial embedding of $G$ is \emph{c-planar} 
if additionally the embedding is combinatorially described as given.
 A \emph{clustered drawing and embedding} of a flat clustered graph $(G,T)$ is a drawing and embedding, respectively,
 of $G$ satisfying (i) and (ii).
In 1995
 Feng, Cohen and Eades~\cite{Feng95,Feng95+} introduced the notion of clustered planarity for clustered graphs, shortly c-planarity, (using, a more general, hierarchical clustering)
as a natural generalization of graph planarity. (Under a different name
Lengauer~\cite{L89} studied a similar concept in 1989.)

\paragraph{Edge contraction and vertex split.}
A \emph{contraction} of an  edge $e=uv$ in a topological graph is an operation that turns
$e$ into a vertex
by moving $v$ along $e$ towards $u$ while dragging all the other edges incident to $v$ along $e$.
Note that by contracting an edge in an even drawing, we obtain again an even drawing.
By a contraction we can introduce multi-edges or loops at the vertices.

We will also often use the following operation which can be thought of as the inverse operation of the edge contraction
in a topological graph.
A \emph{vertex split} in a drawing of a graph $G$ is the operation that replaces a vertex $v$ by two vertices $v'$ and $v''$
drawn in a small neighborhood of $v$ joined by a short crossing free edge so that the neighbors of $v$ are partitioned into two parts
according to whether they are joined with $v'$ or $v''$ in the resulting drawing, the rotations at $v'$ and $v''$ are inherited from the
rotation at $v$, and the new edges are drawn in the small  neighborhood of the edges they correspond to in $G$.

\paragraph{Even drawings.}
A pair of edges in a graph is \emph{independent} if they do not share a vertex.
An edge in a drawing is \emph{even} if it crosses every other edge an even number of times.
An edge in a drawing is \emph{independently even} if it crosses every other non-adjacent edge an even number of times.
A drawing of a graph is \emph{(independently) even} if all edges are (independently) even. Note that an embedding is an even drawing.

\paragraph{Edge-vertex switch.}
In our arguments we use a continuous deformation in order to transform a given drawing into a drawing with desired properties.
Observe that during such transformation of a drawing of a graph
the parity of crossings between a pair of edges is affected only when an edge $e$ passes over a vertex $v$,
in which case we change the parity of crossings of $e$ with all the edges incident to $v$. Let us call such an event an \emph{edge-vertex switch}.

\subsection{Hanani--Tutte}

\label{sec:hanani}

The Hanani--Tutte theorem~\cite{C34,T70} is a classical result that provides an algebraic characterization of planarity with interesting algorithmic consequences~\cite{FKMP15}. 
The (strong) Hanani--Tutte theorem says that a graph is planar as soon as it can be drawn in the plane so that no pair of edges that do not share a vertex
cross an odd number of times.
Moreover, its variant known as the weak Hanani--Tutte theorem~\cite{CN00,PT00,PSS06} states that if we have a drawing $\mathcal{D}$ of a graph $G$ where every pair of edges cross an even number of times then $G$ has an embedding that preserves the cyclic order of edges at vertices from $\mathcal{D}$.
Note that the weak variant does not directly follow from the strong Hanani--Tutte theorem.
For sub-cubic graphs, the weak variant implies the strong variant.

Other variants of the Hanani--Tutte theorem in the plane were proved for $x$-monotone drawings~\cite{FPSS12,PT04_monotone},
partially embedded planar graphs, simultaneously embedded planar graphs~\cite{S12+}, and two--clustered graphs~\cite{FKMP15}.
As for the closed surfaces of genus higher than zero, the weak variant is known to hold in all closed surfaces~\cite{PSS09}, and the strong variant was proved only
for the projective plane~\cite{PSS09c}.
It is an intriguing open problem to decide if the strong Hanani--Tutte theorem holds
for closed surfaces other than the sphere and projective plane.

There is, however, another tightly related line of research on approximability or realizations of maps pioneered by Sieklucki~\cite{S69}, Minc~\cite{M97} and M.~Skopenkov~\cite{S03} that is completely independent from the aforementioned developments.
\cite[Theorem 1.5]{S03} is a weak variant of the Hanani--Tutte theorem for flat cluster graphs with three clusters or cyclic clustered graphs~\cite[Section 6]{FKMP15}.

To prove a strong variant for a closed surface it is enough to prove it for all the minor minimal graphs (see e.g.~\cite{D10} for the definition of a graph minor) not embeddable in the surface. Moreover, it is known that the list of such graphs is finite for every closed surface, see e.g.~\cite[Section 12]{D10}. Thus, proving or disproving the strong Hanani--Tutte theorem on a closed surface boils down to a search for a counterexample among a finite number of graphs. That sounds quite promising, since checking a particular graph is reducible to a finitely many, and not so many, drawings, see e.g.~\cite{SS13}. However, we do not have a complete list of such graphs for any surface besides the sphere and projective plane.

On the positive side, the list of possible minimal counterexamples for each surface was recently narrowed down to vertex two-connected graphs~\cite{SS13}.
See~\cite{S12} for a recent survey on applications of the Hanani--Tutte theorem and related results.

\subsection{Necessary conditions for $x$-boundedness}

\label{sec:crossing}

We present two necessary conditions for the isotopy class $\mathcal{E}$ of  an embedding of $(G,\gamma)$ to contain an $x$-bounded embedding. In Section~\ref{sec:char} we show that the conditions
are, in fact, also sufficient, which  implies Theorem~\ref{thm:main}.
For the remainder of this section we assume that $G$ is 
given by the isotopy class of its embedding $\mathcal{E}$.

In what follows we give an equivalent definition of the one from Secion~\ref{sec:notation} of $i_\mathcal{E}(P_1,P_2)$, \emph{algebraic intersection number}~\cite{CN00}  of a pair of oriented paths $P_1$ and $P_2$ in an isotopy class of an embedding of a graph. This definition is easier to work with. We orient $P_1$ and $P_2$ arbitrarily.
 Let $P$ denote the subgraph of $G$ which is the union of $P_1$ and $P_2$.
We define $cr_{P_1,P_2}(v)=+1$ ($cr_{P_1,P_2}(v)=-1$) if $v$ is a vertex of degree four in $P$ such that the paths $P_1$ and $P_2$ alternate in the rotation at $v$
and at $v$ the path $P_2$ crosses $P_1$ from left to right (right to left) with respect
to the chosen orientations of $P_1$ and $P_2$.
We define $cr_{P_1,P_2}(v)=+1/2$ ($cr_{P_1,P_2}(v)=-1/2$) if $v$ is a vertex of degree three in $P$ such that at $v$ the path
$P_2$ is oriented towards $P_1$ from left, or from $P_1$ to right (towards $P_1$ from right, or from $P_1$ to left) in the direction of $P_1$.
The algebraic intersection number of $P_1$ and $P_2$ is then the sum of $cr_{P_1,P_2}(v)$ over all vertices of degree three and four in $P$.

We extend the notion of algebraic intersection number to oriented walks as follows.
Let $i_\mathcal{E}(W_1,W_2)=\sum cr_{u_1v_1w_1,u_2v_1w_2}(v_1)$, where the sum runs over all pairs $u_1v_1w_1\subseteq W_1$ and $u_2v_1w_2\subseteq W_2$
of oriented sub-paths of $W_1$ and $W_2$, respectively. (Sub-walks of length two in which $u_1=w_1$ or $u_2=w_2$ does not
 have to be considered in the sum, since their contribution towards the algebraic intersection number is zero anyway.)
Note that $i_\mathcal{E}(W_1,W_2)$ is zero for a pair of closed walks.
Indeed, $i_\mathcal{E}(C_1,C_2)=0$ for any pair of closed continuous curves in the plane which can be proved by observing that the statement
is true for a pair of non-intersecting curves and preserved
under a continuous deformation.
Whenever talking about algebraic intersection number of a pair of walks we tacitly assume that the walks are oriented. The actual orientation is not important to us since
in our arguments only the absolute value of the algebraic intersection number matters.

Let $G'\subseteq G$. Let $\max (G')$ and $\min (G')$, respectively, denote the maximal and minimal value of $\gamma(v)$, $v\in V(G')$.

\begin{figure}[t]
  \centering
\centering
{
\includegraphics[scale=0.7]{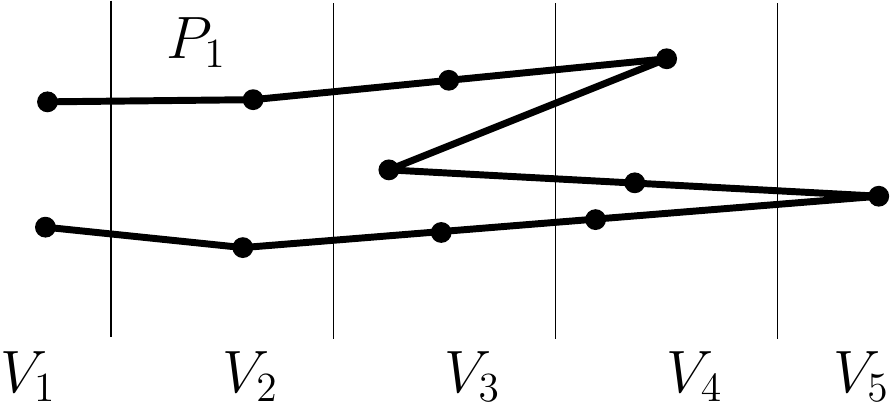}
    	} \hspace{10px}
{ 
\includegraphics[scale=0.7]{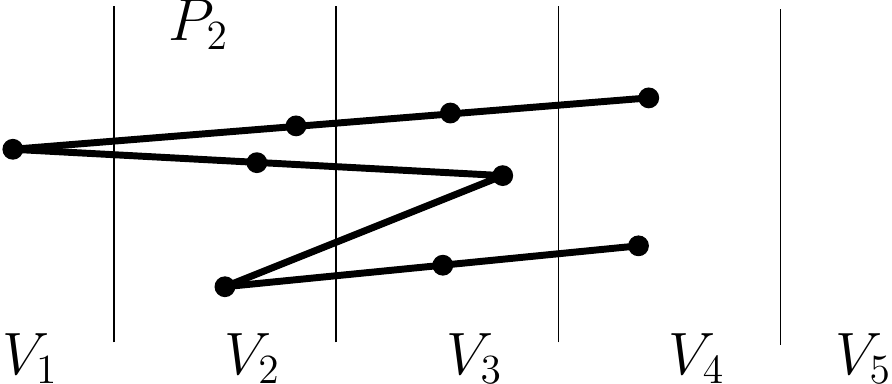}
		}
\caption{A path $P_1$ that is 1-cap (top); and a path $P_2$ that is a 4-cup (bottom).}
\label{fig:cupcup}
\end{figure}

\begin{figure}[t]
  \centering
\centering
{
\includegraphics[scale=0.7]{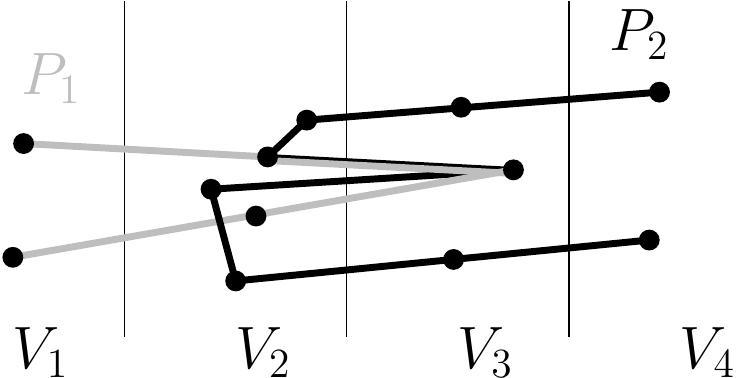}
    	} \hspace{10px}
{
\includegraphics[scale=0.7]{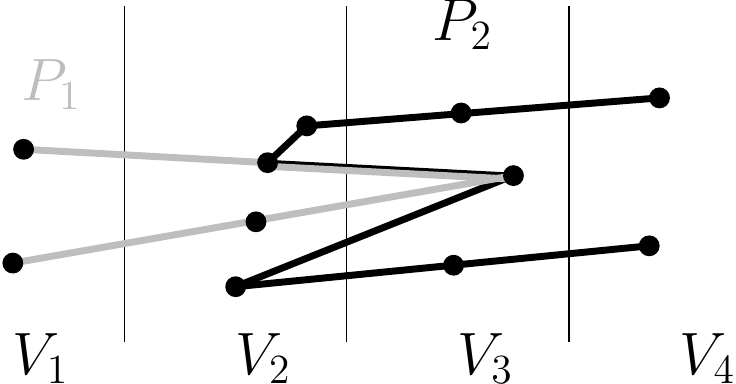}
		}
\caption{An infeasible pair of an 1-cap $P_1$ and a 4-cup  $P_2$ (top); and a feasible pair of an 1-cap $P_1$ and a 4-cup  $P_2$ (bottom).}
\label{fig:cupcup2}
\end{figure}

\paragraph{Definition of an $i$-cap and $i$-cup.}
A path $P$ in $G$ is an \emph{$i$-cap} and \emph{$j$-cup} if for the end vertices $u,v$ of $P$ and all $w\not=u,v$ of $P$
we have $\min (P) = \gamma (u) =\gamma(v)=i\not=\gamma(w)$ and $\max (P) = \gamma (u) =\gamma(v)=j\not=\gamma(w)$, respectively, (see Figure~\ref{fig:cupcup}).
A pair of  an \emph{$i$-cap} $P_1$ and \emph{$j$-cup} $P_2$ is \emph{interleaving} if (i) $\min(P_1)< \min(P_2)\le \max(P_1) < \max(P_2)$; and (ii) $P_1$ and $P_2$ intersect in a path (or a single vertex).
An interleaving pair of an oriented $i$-cap $P_1$ and $j$-cup $P_2$ is \emph{infeasible}, if  $i_\mathcal{E}(P_1,P_2)\not=0$, and \emph{feasible},
otherwise (see Figure~\ref{fig:cupcup2}).
Thus, feasibility does not depend on the orientation.
Note that $i_\mathcal{E}(P_1,P_2)$ can be either $0,1$ or $-1$. Throughout the paper by an infeasible and feasible pair of paths we mean an infeasible and feasible, respectively,  interleaving pair of an {$i$-cap} and {$j$-cup}.


\begin{observation}
\label{obs:crossing}
In $\mathcal{E}$ there does not exist an infeasible interleaving pair $P_1$ and $P_2$ of an $i$-cap and $j$-cup, $i+1<j$.
\end{observation}

As a special case of Observation~\ref{obs:crossing} we obtain the following.

\begin{observation}
\label{obs:alternate}
The incoming and outgoing edges do not alternate at any vertex $v$ of $\overrightarrow{G}$  (defined in Section~\ref{sec:notation}) in the rotation given by $\mathcal{E}$, i.e., the
incoming and outgoing edges incident to $v$ form two disjoint intervals in the rotation at $v$.
\end{observation}

We say that a vertex $v\in V(G)$ is \emph{trapped} in the interior of a cycle $C$
if in $\mathcal{E}$ the vertex $v$ is in the interior of $C$ and we have $\min(C) > \gamma(v)$ or $\gamma(v)> \max (C) $, where
$\max (C)$ and $\min (C)$, respectively, denotes the maximal and minimal label of a vertex of $C$.
A vertex $v$ is trapped if it is trapped in the interior of a cycle.

\begin{observation}
\label{obs:trap}
In $\mathcal{E}$ there does not exist a trapped vertex.
\end{observation}

\subsection{Proof of Lemma~\ref{lemma:intro}}

\begin{proof}
W.l.o.g. we assume that $G$ is connected.
By~\cite[Lemma 2]{F14} we deform the given $x$-bounded embedding into an $x$-bounded drawing
in which every pair of edges cross an even number of times.
In the obtained drawing, we contract every connected component of $G$ induced by vertices with the same $\gamma$ value to a point thereby possibly obtaining loops at vertices. Let $(G',\gamma')$ denote the resulting pair.
Note that in the drawing of $G'$ every pair of edges still cross  an even number of times.
Hence, in $G'$ the loops can be redrawn in the close vicinity of their vertices
thereby making them crossing free without changing the rotation system.
See the proof of Theorem 1 in~\cite{FKMP15} for a more formal treatment of the previous argument.
 By using the corresponding variant of the weak Hanani--Tutte theorem~\cite[Theorem 1]{F14} 
we obtain an $x$-bounded embedding of $(G',\gamma')$ in which $x(e)$ of every non-loop edge  $e\in E$ is injective without changing the rotation at vertices.
Indeed, the {$x$-monotonicity} follows directly from the proof.
The contracted components $C$ can be recovered as follows.
We embed $C$ represented in $G'$ by a vertex $v$ in a close vicinity
of $v$ by Tutte's barycenter algorithm~\cite{tutte1963draw}.
To this end we first sub-divide edges incident to $v$ and un-contract $C$ (which is possible since we did not change
the rotations, and thus, the loops at $v$ are still edge-crossing free).
Let $C'$ denote the union of $C$ with the edges leaving $C$.
Due to sub-divisions of the edges incident to $v$ all the edges leaving $C$ have degree one.
Let $v_1,\ldots, v_k$ denote those degree-one (in $C'$) vertices.
Note that $v_i$'s have degree two in $G'$.
We augment the obtained embedding of $C'$ into an internally triangulated planar graph $C''$ having the outer face 
bounded by the cycle $v_1,\ldots, v_k$.
In the embedding $\mathcal{D}$ obtained by the variant of the weak Hanani--Tutte theorem~\cite[Theorem 1]{F14}
we replace a small disc  neighborhood $D_v$ of $v$ by the straight-line embedding of  $C''$
obtained by an application of Tutte's barycenter algorithm. 
We  assume that $v_i$'s are drawn on the boundary of $D_v$ in $\mathcal{D}$.
In the barycentric embedding of $C''$ the vertices $v_1,\ldots, v_k$ are prescribed
to lie on the boundary of $D_v$ as in $\mathcal{D}$.
By recovering contracted components in $\mathcal{D}$ one by one as above the claim follows.
\end{proof}

\section{Characterization of isotopy  classes containing $x$-bounded embeddings}
\label{sec:char}

In this section we prove our characterization of  isotopy
classes $\mathcal{E}$ of $G$ containing an $x$-bounded embedding w.r.t. $(G,\gamma)$
by reducing a general instance of to a normalized one.

\begin{theorem}
\label{thm:characterization}
 The isotopy class $\mathcal{E}$ of $G$ contains an $x$-bounded embedding  w.r.t. $(G,\gamma)$ if and only if $\mathcal{E}$ does not contain an infeasible interleaving pair of paths, or a trapped vertex.
\end{theorem}

Before we turn to the proof of Theorem~\ref{thm:characterization} we discuss its
relation to  Theorem~\ref{thm:main}.
The condition that $\mathcal{E}$ does not contain
a trapped vertex is an equivalent reformulation
of the condition that $i_\mathcal{E}(D,v)=0$,
where $D$ is a union of faces with a disc as a support, if $\gamma(\partial (D)=C) \cap \gamma(v)=
\emptyset$.
Regarding the condition for  pairs of paths, Theorem~\ref{thm:characterization} seems to be stronger than Theorem~\ref{thm:main}
due to a more restricted condition on pairs of paths we consider.
However, the strengthening is not significant, since it can be easily shown that forbidding an infeasible interleaving pair of paths
and trapped vertices renders the hypothesis of the ``if'' part of Theorem~\ref{thm:main} satisfied, and thus,
we get its equivalence with Theorem~\ref{thm:characterization}.

\begin{wrapfigure}{r}{.3\textwidth}
\centering
\includegraphics[scale=0.7]{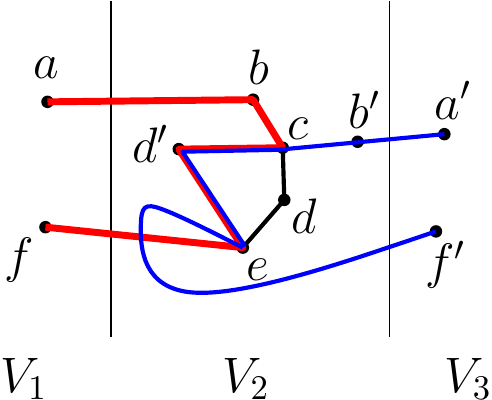}
\caption{Replacing the portion of $P_1'=abcdef$ with the portion of $P_2'=a'b'cd'ef'$ on the cycle $C=cded'$.}
\label{fig:ffEx}
\end{wrapfigure}

Indeed, if a pair of intersecting paths $P_1$ and $P_2$ satisfies $\gamma(P_1) \cap \gamma(\partial P_2) = \emptyset$ and $\gamma(\partial P_1) \cap \gamma(P_2) = \emptyset$,
there exist sub-paths $P_1'$ of $P_1$ and $P_2'$ of $P_2$ such that $i_\mathcal{E}(P_1,P_2)=i_\mathcal{E}(P_1',P_2')$ that either form an interleaving pair
or do not form an interleaving pair only because they do not intersect in a path.
In the latter, no end vertex of $P_1$ or $P_2$ is contained in the interior
of a cycle in $P_1' \cup P_2'$ due to the non-existence of trapped vertices.
Let $W_2$ be a walk obtained from $P_2'$ by replacing its portion on a cycle $C$
contained in $P_1' \cup P_2'$, such  that $P_1' \cap C$ is a path, with the portion of $P_1'$ for every such cycle (see Figure~\ref{fig:ffEx}). Let $P_2''$ denote the path in $W_2$ connecting
its end vertices.
We have $i_\mathcal{E}(P_1,P_2)=i_\mathcal{E}(P_1',P_2'')$, and $P_1'$ and $P_2''$ form an interleaving pair.
Hence, we just proved the following.

\begin{lemma}
\label{lemma:conjChar}
Given that $(G,\gamma)$ is free of trapped vertices,
if a pair of intersecting paths $P_1$ and $P_2$ in $\mathcal{E}$ satisfies $\gamma(P_1) \cap \gamma(\partial P_2) = \emptyset$ and $\gamma(\partial P_1) \cap \gamma(P_2) = \emptyset$\footnote{In the case of paths the boundary operator $\partial$ returns the end vertices.},
there exist sub-paths $P_1'$ of $P_1$ and $P_2'$ of $P_2$ such that $i_\mathcal{E}(P_1,P_2)=i_\mathcal{E}(P_1',P_2'')$, where $P_2''\subset P_1' \cup P_2'$ is constructed as above, forming an interleaving pair.
\end{lemma}

\begin{proof}[proof of Theorem~\ref{thm:characterization}]
The proof is inspired by the work of Minc~\cite{M97} and M.~Skopenkov~\cite{S03}.
By sub-dividing edges of $G$ we (tacitly) assume  $|\gamma(u)-\gamma(v)|\leq 1$, for every edge $uv\in E(G)$.
We proceed by the induction on the number of clusters $c$ and  \\
$\lambda:=\lambda(G,\gamma)=\sum_{V_i}\sum_{C_j}(|V(C_j)|-1)$, where the inner sum is over the connected components $C_j$ induced by $V_i$, in this order.

Suppose that $\lambda>0$.
It follows that we have an edge $e$ in $E(G)$ between two vertices $u,v$ with the same $\gamma$ value.
We contract $e$ into a vertex $w$ in an embedding of $G$ from the given isotopy class thereby  decreasing $\lambda$. We put $\gamma'(w):=\gamma(u)$ and $\gamma'(z):=\gamma(z)$ for every other vertex of $G$. Let $(G',\gamma)$ denote the obtained pair. The resulting drawing is still an embedding but we could introduced a loop at $w$ by the contraction.
However, we did not introduce a trapped vertex or an infeasible interleaving pair of paths. In particular, if there exists a loop incident to $w$ it contains only vertices with the same $\gamma$ value as $w$.
We delete such loops together with its interior.
We apply the induction hypothesis on the obtained pair $(G'',\gamma'')$ with the isotopy class of its obtained embedding thereby obtaining an $x$-bounded embedding of $(G'',\gamma'')$.
In the $x$-bounded embedding of $G'$ we re-introduce deleted loops  with their interiors at the same position in the rotation at $w$.
Then by splitting $w$ into $e$ we obtain a desired $x$-bounded embedding of $(G,\gamma)$.


Hence, suppose that $\lambda=0$. Let $v_i\in V_i$, $i=0,\ldots, c-1$, such that $v_i$ is joined by an edge with a vertex $u_{i+1}\in V_{i+1}$  and a vertex $w_{i-1}\in V_{i-1}$.   We apply a vertex split to every such $v_i$ thereby obtaining a pair of new vertices $v_i'$ and $v_i''$ joined by an edge such that $\gamma(v_i')=\gamma(v_i''):=\gamma(v_i)=i$.
The vertex $v_i'$ is joined by an edge with neighbors of $v_i$ in $V_{i-1}$ and $v_i''$ is joined by an edge with neighbors of $v_i$ in $V_{i+1}$.
Let $(G',\gamma')$ denote the obtained pair.
We assume that the rotation system of $G'$ is such that by contracting the edges introduced by the splits we obtain an embedding in the given isotopy class of $G'$.
Since we have no infeasible interleaving pair in $(G,\gamma)$ this is possible.
Let $G_i$, for $i=1,\ldots, c-1$,  denote the sub-graph of $G'$ induced by the edges in $G'$ between $V_i'$ and $V_{i-1}'$, i.e,
the sets of vertices with $\gamma'$ value $i$ and $i-1$.
Note that $G_i$ is an induced sub-graph of $G'$.
Let $(G',\gamma'')$ be such that the image of $\gamma''$ has $c-1$ different values and $\gamma''(V(G_1))=1,\ldots, \gamma''(V(G_{c-1}))=c-1$.
We show that $(G',\gamma'')$  contains neither an interleaving pair of paths nor a trapped vertex.

For the sake of contradiction, suppose that $(G',\gamma'')$ contains an infeasible interleaving pair of paths, an $i$-cap $P_1'$ and $j$-cup $P_2'$. Note that we can assume that $P_1'$ ends in a vertex from $V(G_i)\cap V_{i-1}'$ and $P_2'$ ends in a vertex from $V(G_j)\cap V_{j}'$. Then we see that $P_1'$ and $P_2'$ yield an infeasible interleaving pair in $(G,\gamma)$ (contradiction). 
 Similarly, we can argue about trapped vertices. Hence, by the induction hypothesis $(G',\gamma'')$ admits an $x$-bounded embedding.
Let $\mathcal{E}$ denote the corresponding $x$-bounded embedding.

Note that $G_i$, for $i=1,\ldots, c-1$, is a bipartite graph with partitions 
$V_{i-1}'\cap V(G_i)$ and $V_i' \cap V(G_{i})$, both inducing an independent set.
We showed that every c-graph with two-clusters, both inducing an independent set, is c-planar~\cite{FKMP15}. Moreover, an arbitrary isotopy class can be chosen for the corresponding embedding.
Note that the restriction of the $x$-bounded embedding of $(G',\gamma'')$ to $G_i$ has vertices joining $G_i$ with $G'\setminus G_i$ on the outer face $f_i$ of $G'$, and in the facial walk of $f_i$ wedges containing the edges between $V(G_i)$ and $G_{i+1}$ do not alternate with wedges  containing the edges between $V(G_i)$ and $G_{i-2}$. For otherwise we would obtain an infeasible interleaving pair in $(G',\gamma'')$.
Hence,   such clustered embeddings of $G_i$, for $i=1,\ldots, c-1$, can be put next to each other and connected by edges so as to obtain an embedding of $G'$ in the given isotopy class.  It follows that if we put $\gamma'''(V_{i-1}\cap V(G_i)):=2i-2$ and 
 $\gamma'''(V_{i}\cap V(G_i)):=2i-1$,  the obtained embedding of $G'$ can be deformed to obtain an $x$-bounded  embedding  of $(G',\gamma''')$.
The corresponding $x$-bounded embedding of $(G',\gamma''')$ is turned into an $x$-bounded embedding of  $(G,\gamma)$  by edge contractions and re-scaling $\gamma'''$, and this concludes the proof. 
 \end{proof}

\section{Corollaries of the characterization}

\subsection{The variant of the weak Hanani--Tutte theorem for $x$-bounded drawings}
\label{sec:linearly}

In this section we prove the weak Hanani-Tutte theorem for $x$-bounded drawings, Theorem~\ref{thm:linearly}.

Given a drawing of a graph $G$ where every pair of edges cross an even number of times, by the weak Hanani-Tutte theorem~\cite{CN00,PT00,PSS06},
we can obtain an embedding of $G$ with the same rotation system, and hence, the facial structure of an embedding of $G$ is already present in an even drawing. This allows us to speak about faces in an even drawing of $G$. Hence, a face in an
even drawing of $G$ is the walk bounding the corresponding face in the embedding of $G$ with the same rotation system.

 A face $f$ in an even drawing corresponds to a closed (possibly self-crossing)
curve $C_f$ traversing the edges of the defining walk of $f$ in a close vicinity of its edges without crossing an edge that is being traversed, i.e,
$C_f$ never switches to the other side of an edge it follows.
An \emph{inner face} in an even drawing of $G$ is a face for which all the vertices of $G$ except those incident to $f$ are outside of $C_f$.
Similarly, an \emph{outer face} in an even drawing of $G$ is a face such that all the vertices of $G$ except those incident to $f$ are inside of $C_f$.
Note that by the weak Hanani--Tutte theorem every face is either an inner face or an outer face.
Unlike in the case of an embedding (in the plane), in an even drawing the outer face might not be unique. Nevertheless, an outer face always exists in an even drawing of a graph in the plane.

\begin{lemma}
\label{obs:outer-face}
Every even drawing of a connected graph $G$ in the plane has an odd number of outer faces.
\end{lemma}

\begin{proof}
Refer to Figure~\ref{fig:obserd5}.
By successively contracting every edge in an even drawing of $G$ we obtain a vertex $v$ with a bouquet of loops, see e.g., the proof of ~\cite[Theorem 1.1]{PSS06}.
Let  $\mathcal{D}$ be the obtained drawing whose underlying abstract graph is not  simple
unless it is edgeless.
Let us treat $\mathcal{D}$ as an even drawing. Thus, we obtain the facial structure in $\mathcal{D}$
by traversing walks consisting of loops at $v$.
Every loop $l$ at $v$ corresponds to a cyclic interval in the rotation at $v$ containing
the end pieces of edges that are in a close neighborhood at $v$ contained inside
$l$. By treating every walk in $\mathcal{D}$ as a walk along cyclic intervals of the loops it traverses we define the \emph{winding number} of a face in $\mathcal{D}$ as
the number of times we walk around $v$ when traversing the intervals of its walk.
The winding number can be positive or negative depending on the sense of the traversal.

\begin{figure}[htp]
\centering
\includegraphics[scale=0.7]{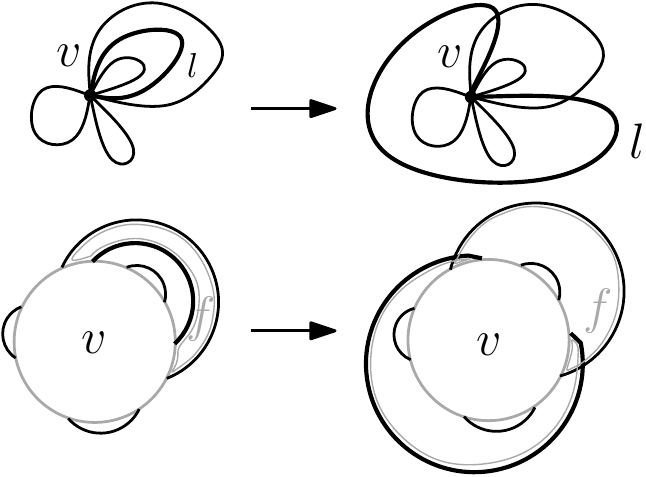}
\caption{ A bouquet of loops at $v$ and its modification after we pull a loop  $l$ over $v$ (top). The cyclic intervals
corresponding to the loops in the rotation at $v$ before and after we pull $l$ over $v$ (bottom).}
\label{fig:obserd5}
\end{figure}

Note that the outer faces in $\mathcal{D}$ are those faces $f$ whose corresponding walks wind around $v$ an odd number of times. This follows because whenever we visit $v$ during a walk of $f$ winding an odd number of times around $v$ the corresponding position in the rotation at $v$ is contained inside of an even number of loops of the walk, and hence outside of $C_f$.

By pulling a loop $l$ over $v$ we flip the cyclic interval in its rotation
that corresponds to the inside of $l$. It follows that we change the winding number of both facial walks that $l$ participates in by one.
Hence, we do not change the parity of the total number of outer faces in $\mathcal{D}$.
Since a crossing free drawing of $G$ has an odd number of outer faces the lemma follows.
\end{proof}

Given an even $x$-bounded drawing of  $(G,\gamma)$ we can associate it with the isotopy class $\mathcal{E}$ of a corresponding embedding of $G$. Note that by the connectivity of $G$ the interval spanned by the $\gamma$ values of the vertices
incident to the outer face $f$ in $\mathcal{E}$,
which can be chosen by Lemma~\ref{obs:outer-face}, span all the values $\gamma$.
By Theorem~\ref{thm:characterization} it is enough to prove that $\mathcal{E}$ does not contain an infeasible pair of paths
or a vertex trapped in the interior of a cycle. However, due to evenness of the given drawing of $G$ both of these
forbidden substructures would introduce a pair of cycles crossing an odd number of times (contradiction).
In order to rule out the existence of a trapped vertex we use the fact that
the boundary of the outer face $f$ spans all the values of $\gamma$. If a vertex $v$ is trapped in the interior of  a cycle $C$ then by the connectedness of $G$
we can join $C$ with $v$ by a path $P$ of $G$. By the evenness of the drawing it follows that the end piece of $P$ at $C$
in our drawing start outside of $C$. On the other hand, a path connecting $C$ with any vertex on the outer face $f$ must
also start at $C$ outside of $C$, since the boundary of $f$ spans all the values of $\gamma$.
 Thus, $v$ cannot be trapped, since the rotation system from the even drawing is preserved
in the embedding.

\subsection{Strip clustered graphs and c-planarity}
\label{sec:strip}

A \emph{clustered graph}\footnote{This type of clustered graphs is usually called flat clustered graph in the graph drawing literature. We choose this simplified notation in order not to overburden the reader with unnecessary notation.}  is an ordered pair $(G,T)$, where $G$ is a graph, and
$T=\{V_i| i=1,\ldots, k\}$ is a partition of the vertex set of $G$ into $k$ parts. We call the sets $V_i$ clusters.
A drawing of a \emph{clustered graph} $(G,T)$ is clustered if vertices in $V_i$ are drawn inside
a topological disc $D_i$ for each $i$ such that $D_i\cap D_j=\emptyset$ and every edge of $G$
intersects the boundary of every disc $D_i$ at most once.
We use the term ``cluster $V_i$'' also when referring to the topological disc $D_i$ containing the vertices in $V_i$.
A clustered graph $(G,T)$ is \emph{clustered planar} (or briefly \emph{c-planar}) if $(G,T)$ has a clustered embedding.

A strip clustered graph is a concept introduced recently
 by Angelini et al.~\cite{ADDF13+}\footnote{The author was interested in this planarity variant independently prior
  to the publication of~\cite{ADDF13+} and adopted
 the notation introduced therein.} For  convenience  we slightly alter their definition and define ``strip clustered graphs''
as  ``proper'' instances of ``strip planarity'' in~\cite{ADDF13+}.
In the present paper we are primarily concerned with the following subclass of clustered graphs.
A \emph{clustered graph} $(G,T)$ is \emph{strip clustered}
 if $G=\left(V_1\cup \ldots \cup V_k, E\subseteq \bigcup_i{V_i \cup V_{i+1} \choose 2}\right)$, i.e., the edges
 in $G$ are either contained inside a part or join vertices in two consecutive parts.
A drawing of a strip clustered graph $(G,T)$ in the plane is \emph{strip clustered} if $i<x(v_i)<i+1$ for
all $v_i\in V_i$, and every line of the form $x=i$, $i\in \mathbb{N}$, intersects every edge at most once.
 Thus, strip clustered drawings constitute a restricted class
of clustered drawings.
We use the term ``cluster $V_i$'' also when referring to the vertical strip containing the vertices in $V_i$.
A strip clustered graph $(G,T)$ is \emph{strip planar}
if $(G,T)$ has a strip clustered embedding in the plane.
Note that if we define $(G,\gamma)$, so that $\gamma(v)=i$ for $v\in V_i$, a strip clustered
drawing is $x$-bounded. Thus, Theorem~\ref{thm:AlgTreeXBounded} implies an efficient
algorithm for strip planarity testing.

\begin{lemma}
\label{lemma:strip3}
The problem of strip planarity testing is reducible in linear time
to the problem of c-planarity testing in the case of flat clustered graphs with three clusters.
\end{lemma}

\begin{proof}
Given an instance of $(G,T)$ of strip clustered graph we construct a clustered graph $(G,T')$ with three clusters $V_0', V_1'$ and  $V_2'$ as follows. We put $T'=\{V_j'|  \ V_j'=\bigcup_{i; \ i \mod 3=j}V_i  \}$. Note that without loss of generality
we can assume that in a drawing of $(G,T')$ the clusters are drawn as regions bounded by a pair of rays emanating from the origin. By the inverse of a projective transformation taking the origin to the vertical infinity we can also assume that the same is true for a drawing of $(G,T)$.
 Notice that such clustered embedding of $(G,T)$ can be continuously deformed by a rotational
transformation of the form $(\phi,r)\rightarrow (k\phi, c\phi \cdot r)$ for appropriately chosen $k,c>1$, which is expressed in polar coordinates,
 so that we obtain a clustered embedding of $(G,T')$.
 We remark that  $(x,y)$ in Cartesian coordinates corresponds to $(\phi,\sqrt{x^2+y^2})$ such that $\sin \phi=\frac{y}{\sqrt{x^2+y^2}}$
and $\cos \phi=\frac{x}{\sqrt{x^2+y^2}}$ in polar coordinates.
On the other hand, it is not hard to see that if $(G,T')$ is c-planar then there exists a clustered embedding of $(G,T')$ with the following property.
For each $i=0,1,2$ and $j$ the vertices of $V_i'$ belonging to $V_j$
and the parts of their adjacent edges in the region representing $V_j'$ belong to a topological disc
 $D_j$ such that $D_j\cap D_{j'}=\emptyset$ for $j\not=j'$ fully contained in this region.
To this end we proceed as follows.
Let $E_i=E[V_{i-1},V_{i}]$ denote the edges in $G$ between $V_{i-1}$ and $V_{i}$.
Let $r_i$ denote the ray emanating from the origin that separates $V_{i-1}'$ from $V_{i}'$.
Given a clustered drawing of $(G,T')$, $p_e(\mathcal{E})$, for $e\in E_i$, is the intersection point of $e$
with the ray $r_{i \mod 3}$. Let $p$ denote the origin. Let $|pq|$ for a pair of points in the plane denote the Euclidean distance between $p$ and $q$.
Recall that $G$ has clusters $V_1,\ldots, V_k$. We obtain a desired embedding of $(G,T')$ inductively as $\mathcal{E}_{k}$ starting with
$\mathcal{E}_4$.
For $\mathcal{E}_i=(G,T')$, $i=5,\ldots, k$, we maintain the following invariant.
For each $j, \ 5\leq j\leq i$, we have $$\max_{e\in E[V_{j-4},V_{j-3}]}|pp_e| < \min_{e\in E[V_{j-1},V_{j}]}|pp_e| \  \ (*).$$
Let $\mathcal{E}$ denote a clustered embedding of $(G,T')$.
We start with a clustered embedding of $\mathcal{E}_4$ of $(G[V_1\cup V_2\cup V_3\cup V_4],T')$ inherited from $\mathcal{E}$.  In the $i^{\mathrm{th}}$ step of the induction we extend
$\mathcal{E}_{i-1}$ of $(G[V_1\cup V_2\ldots \cup V_{i-1}],T')$ inside the wedge corresponding to $V_{i-1 \mod 3}'$ and $V_{i \mod 3}'$
 thereby obtaining an embedding $\mathcal{E}_{i}$ of
$(G[V_1\cup V_2\ldots \cup V_{i}],T')$ so that the resulting embedding $\mathcal{E}_{i}$
is still clustered, and $(*)$ is satisfied. Since by induction hypothesis we have
$\max_{e\in E[V_{i-3j-2},V_{i-3j-1}]}|pp_e| < \min_{e\in E[V_{i-2},V_{i-1}]}|pp_e|$, for all possible $j$,
in $\mathcal{E}_{i-1}$ we have $G[V_{i-1}]$ drawn in the outer face of  $G[V_1 \cup \ldots \cup V_{i-2}]$.
Thus, we can extend the embedding of $\mathcal{E}_{i-1}$ into $\mathcal{E}_{i}$ in which all the edges
of $E_i$ cross $r_{i \mod 3}$ in the same order as in $\mathcal{E}$ while maintaining the invariant $(*)$
and the rotation system inherited from $\mathcal{E}$.
The obtained embedding $\mathcal{E}_{k}$ of $(G,T')$  can be easily transformed into a strip clustered embedding.

 Thus, $(G,T)$ is strip planar if and only if $(G,T')$ is c-planar.
\end{proof}

If $G$ is a tree also the converse of Lemma~\ref{lemma:strip3} is true. In other words,
 given an instance of clustered tree $(G,T)$ with three clusters $V_0, V_1$ and  $V_2$ we can easily construct a strip clustered tree $(G,T')$
with the same underlying abstract graph
such that $(G,T')$ is strip planar if and only if $(G,T)$ is c-planar. Indeed, the desired equivalent instance
 is obtained by partitioning the vertex set of $G$ into clusters thereby obtaining $(G,T'=\{V_i'|\ i\in I\subset \mathbb{N}\})$ as follows.
In the base case, pick an arbitrary vertex $v$ from a non-empty cluster $V_i$ of $G$ into $V_{i}'$,
and no vertex is processed.

In the inductive step we pick an unprocessed vertex $u$ that was already put into a set $V_j'$ for some $j$.
 We put neighbors of
$u$ in $V_{j \mod 3}$ into $V_j'$, neighbors in $V_{j+1 \mod 3}$ into $V_{j+1}'$,
 and neighbors of $v$ in $V_{j-1 \mod 3}$ into $V_{j-1}'$. Then we mark $u$ as processed.
 Since $G$ is a tree, the partition $T'$ is well defined.
Now, the argument of Lemma~\ref{lemma:strip3} gives us the following.

\begin{lemma}
\label{lemma:striptree}
The problem of c-planarity testing in the case of flat clustered graphs with three clusters is reducible in linear time
to the strip planarity testing if the underlying abstract graph is a tree.
\end{lemma}

\subsection{The variant of the Hanani-Tutte theorem for $x$-bounded drawings and 3-connected graphs}

\label{sec:linearlyStrong}

In this section we prove the Hanani-Tutte theorem  for $x$-bounded drawings if the underlying abstract graph is three connected, Theorem~\ref{thm:linearlyStrong}.

First, we prove a lemma that allows us to get rid of odd crossing pairs by doing only local redrawings and vertex splits.
A drawing of a graph $G$ is obtained from the given drawing of $G$ by \emph{redrawing edges locally at vertices}
if the resulting drawing of $G$ differs from the given one only in small pairwise disjoint neighborhoods of  vertices not containing any other
vertex. The proof of the following lemma is inspired by the proof of~\cite[Theorem 3.1]{PSS06}.

\begin{lemma}
\label{lemma:removeOdd}
Let $G$ denote a subdivision of a vertex three-connected graph drawn in the plane so that every pair of non-adjacent edges cross an even number of times.
We can turn the drawing of $G$ into an even drawing by a finite sequence of local redrawings of edges at vertices and vertex splits.
\end{lemma}

\begin{proof}
We process cycles in $G$ containing an edge crossed by one of its adjacent edges an
odd number of times one by one until no such cycle exists.
Let $C$ denote a cycle of $G$. By local redrawings at the vertices of $C$ we obtain a drawing of $G$, where every edge of $C$ crosses
every other edge an even number of times. Let $v$ denote a vertex of $C$.

First, suppose that every edge incident to $v$ and starting inside of $C$ crosses every edge incident to $v$ and starting outside of $C$ an even number of times.
In this case we perform at most two subsequent vertex splits.
If there exists at least two edges starting at $v$ inside (outside) of $C$, we split $v$ into two vertices $v'$ and $v''$ joined by a very short crossing free edge so that $v'$ is incident to the neighbors of $v$ formerly joined with
$v$ by edges starting inside (outside) of $C$, and $v''$ is incident to the rest of the neighbors of $v$.
Thus, $v''$ replaces $v$ on $C$.
Notice that by splitting we maintain the property of the drawing to be independently even,
and the property of our graph to be three-connected. Moreover, all the edges incident to the resulting vertex $v''$ of degree three or four cross one another an even number of times. Hence, no edge of $C$ will ever be crossed by another edge an odd number of times, after
we apply appropriate vertex splits at every vertex of $C$.

Second, we show that there does not exist a vertex $v$ incident to $C$ so that an edge $vu$ starting inside of $C$ crosses an edge $vw$ starting outside of $C$
an odd number of times. Since $G$ is a subdivision of a vertex three-connected graph, there exist two distinct vertices $u'$ and $w'$ of $C$ different from $v$ such that $u'$ and $w'$, respectively, is connected with $u$ and $w$ by a path internally disjoint from $C$. Let $uP_1u'$ and $wP_2w'$, respectively, denote this path. Note that $u$ can coincide with $u'$ and $w$ can coincide with $w'$.
Let $vP_3u'$ denote the path contained in $C$ no passing through $w'$. Let $C'$ denote the cycle obtained by concatenation of $P_1$, $P_3$, and $vu$.
Let $C''$ denote the cycle obtained by concatenating $P_2$ and the portion of $C$ between $w'$ and $v$ not containing $u'$.
Since $vw$ and $vu$ cross an odd number of times and all the other pairs of edges $e\in E(C')$ and $f\in E(C'')$ cross an even number of times,
the edges of $C'$ and $C''$ cross an odd number of times. It follows that their corresponding
curves cross an odd number of times (contradiction).

Notice that by vertex splits we decrease the value of the function $\sum_{v\in V(G)}deg^3(v)$ whose value is always non-negative.
Hence, after a finite number of vertex splits we turn $G$ into an even drawing of a new graph $G'$.
\end{proof}

We turn to the actual proof of Theorem~\ref{thm:linearlyStrong}.

We apply Lemma~\ref{lemma:removeOdd} to $(G,\gamma)$ thereby obtaining  $(G',\gamma')$, where each vertex obtained by a vertex split,
has the $\gamma'$ value as its parental vertex.
By applying Theorem~\ref{thm:linearly} to $(G',\gamma')$ we obtain a clustered embedding of $(G',\gamma')$. Finally, we contract the pairs of vertices
obtained by vertex splits in order to obtain an $x$-bounded embedding of $(G,\gamma)$.

\section{Trees}

\label{sec:tree}

In this section we give an algorithm proving Theorem~\ref{thm:AlgTreeXBounded}.

In order to make the present section easier do digest, as a warm-up we give an algorithm in the case, when $G$ is a subdivided star.
Then we show that a slightly more involved algorithm based on the same idea also works for general trees. 
Throughout the present section we (tacitly) assume the following $$ (*) \ \ \  |\gamma(u)-\gamma(v)|\leq 1$$ for every edge $uv\in E(G)$. Thus, we can think of proving
the result for strip clustered graphs, which is how we thought
about it originally.

\subsection{Subdivided stars}

\label{sec:star}

In the sequel $G=(V,E)$ is a subdivided star. Thus, $G$ is a connected graph that  contains a special vertex $v$, \emph{the center of the star},
of an arbitrary degree and all the other vertices in $G$ are either of degree one or two.
The assumption $(*)$ can be imposed without loss of generality.
Indeed, by subdividing every edge $uv$ by vertices $w$ for
which $\gamma(u)<\gamma(w)<\gamma(v)$ we do not change the embeddability.

Recall that $\max (G')$ and $\min (G')$, respectively, denote the maximal and minimal value of $\gamma(v)$, $v\in V(G')$, and that
a path $P$ in $G$ is an \emph{$i$-cap} and \emph{$i$-cup}, respectively, if for the end vertices $u$ and $v$ of $P$ and all $w\not=u,v$ of $P$
we have $\min (P) = \gamma (u) =\gamma(v)=i\not=\gamma(w)$ and $\max (P) = \gamma (u) =\gamma(v)=i\not=\gamma(w)$.

The following lemma is a direct consequence of our characterization stated in Theorem~\ref{thm:characterization}.

\begin{lemma}
\label{lemma:cap-cupG}
Let us fix a rotation at $v$, and thus, an embedding of $G$.
Suppose that every interleaving pair of an $i$-cap $P_1$ and $j$-cup $P_2$ in $G$ containing $v$ in their interiors is feasible in the fixed embedding of $G$.
 Then $(G,\gamma)$ admits an $x$-bounded embedding and in a corresponding  embedding of $(G,T)$ the rotation at $v$ is preserved.
\end{lemma}

In what follows we show how to use Lemma~\ref{lemma:cap-cupG} for a polynomial-time $x$-bounded embeddability testing if the underlying abstract graph
is a subdivided star. The algorithm is based on testing in  polynomial time  whether the columns  of a  0--1 matrix can be ordered so that, in every row, either the ones or the zeros are
consecutive.
We, in fact, consider matrices containing $0,1$ and also an ambiguous symbol $*$. A matrix containing 0,1 and $*$ as its elements has the \emph{circular-ones property}
if there exists a permutation of its columns such that in every row, either the ones or the zeros are
consecutive among undeleted symbols after we delete all $*$.  Then each row in the matrix corresponds to a constraint imposed on the rotation at $v$ by Lemma~\ref{lemma:cap-cupG}
simultaneously for many pairs of paths.

By Lemma~\ref{lemma:cap-cupG} it is enough to decide if there exists a rotation at $v$
so that every interleaving pair of an $s$-cap $P_1$ and $b$-cup $P_2$ meeting at $v$ is feasible. Note that if either $P_1$ or $P_2$ does not contain
$v$ in its interior the corresponding pair is feasible.
%
An interleaving pair $P_1$ and $P_2$ passing through $v$ restricts the set of all rotations at $v$ in an $x$-bounded embedding of $(G,\gamma)$.
Namely, if $e_i$ and $f_i$ are edges incident to $P_i$ at $v$ then in an $x$-bounded embedding of $(G,\gamma)$ in the rotation
at $v$ the edges $e_1,f_1$ do not alternate with the edges $e_2,f_2$, i.e.,
 $e_1$ and $f_1$ are consecutive when we  restrict the rotation to
  $e_1,f_1,e_2,f_2$. We denote such a restriction by $\{e_1f_1\}\{e_2f_2\}$.
Given a cyclic order $\mathcal{O}$ of edges incident to $v$, we can interpret $\{e_1f_1\}\{e_2f_2\}$ as a Boolean predicate
which is ``true'' if and only if $e_1,f_1$ do not alternate with the edges $e_2,f_2$ in $\mathcal{O}$.
Of course, for a given cyclic order we have $\{ab\}\{cd\}$ if and only if $\{cd\}\{ab\}$,
and $\{ab\}\{cd\}$ if and only if $\{ba\}\{cd\}$.
Then our task is to decide in  polynomial time if the rotation at $v$ can be chosen so that
 the predicates $\{e_1f_1\}\{e_2f_2\}$ of all the interleaving pairs $P_1$ and $P_2$ are ``true''.
 The problem of finding a cyclic ordering satisfying a given set
  of Boolean predicates of the form $\{e_1f_1\}\{e_2f_2\}$  is $\cNP$-complete in general, since the problem of total ordering~\cite{O79} can be easily  reduced to it in  polynomial time.
  However, in our case the instances satisfy the following structural properties making the problem tractable
  (as we see later).


\begin{observation}
\label{obs:alegebera}
If $\{ab\}\{cd\}$ is false and $\{ab\}\{de\}$ (is true) then $\{ab\}\{ce\}$ is false.
\end{observation}

The restriction on  rotations at $v$ by the pair of an $s$-cap $P_1$ and $b$-cup $P_2$ is \emph{witnessed} by
an ordered pair  $(s,b)$, where $s<b$. We treat such pair as an interval in $\mathbb{N}$. \\
Let $I=\{(s,b)| \ (s,b) \mathrm{ \ witness \ a \ restriction \  on \ rotations \ at} \ v \mathrm{\ by  \ an  \ interleaving \ pair \
of \ paths} \}$.

\begin{observation}
\label{obs:growth}
If an $s$-cap $P$ contains $v$ then $P$ contains an $s'$-cap $P'$ containing $v$ as a sub-path for every $s'$ such that $s<s'<\gamma(v)$.
Similarly, if a $b$-cup $P$ contains $v$ then $P$ contains a $b'$-cup $P'$ containing $v$ as a sub-path for every $b'$ such that $\gamma(v)<b'<b$.
\end{observation}

\begin{observation}
\label{obs:interleave}
Let $s<s'<b<b'$, $s,s',b,b'\in \mathbb{N}$. If the set $I$ contains both $(s,b)$ and $(s',b')$,
it also contains $(s,b')$ and $(s',b)$.
\end{observation}

We would like to reduce the question of determining if we can choose a rotation at $v$ making all the interleaving pairs feasible
to the following problem.
Let $S=\{e_1,\ldots, e_n\}$ of $n$ elements (corresponding to the edges incident to $v$).
Let $\mathcal{S}'=\{L_i',R_i'|\ i=1,\ldots \}$ of  polynomial size in $n$ such that $R_i',L_i'\subseteq S$ and  
$|L_i'|,|R_i'|\ge 2, \ L_{i+1}' \cup R_{i+1}' \subseteq L_i' \cup R_i'$.
  Can we cyclically order $S$ so that
both $R_i'$ and $L_i'$, for every  $R_i', L_i' \in \mathcal{S}'$, appear consecutively, when restricting the order to $R_i' \cup L_i'$?
Once we accomplish the reduction, we end up with the problem of testing the circular-ones property on matrices containing $0,1$ and $*$ as elements,
where each $*$ has only  $*$ symbol underneath. This problem is solvable in  polynomial time as we will see later.
We construct an instance for this problem which is a matrix $M=(m_{ij})$ as follows. The $i^\mathrm{th}$ row of $M$ corresponds to the pair $L_i'$ and $R_i'$ and each
column corresponds to an element of $S$.
For each pair $L_i',R_i'$ we have $m_{ij}=0$ if $j\in L_i'$, $m_{ij}=1$ if $j\in R_i'$, and $m_{ij}=*$ otherwise.
Note that our desired condition on $\mathcal{S}'$ implies that in $M$ each $*$ has only $*$ symbols underneath.
The equivalence of both problems is obvious.


In order to reduce our problem of deciding if a ``good'' rotation at $v$ exists, we first linearly
order intervals in $I$.
Let $(s_0,b_0)\in I$ be inclusion-wise minimal interval such that $s_0$ is the biggest and similarly let $(s_1,b_1)\in I$ be inclusion-wise minimal
 such that $b_1$ is the smallest one.
By  Observation~\ref{obs:interleave} we have $s_0=s_1$ and $b_0=b_1$.
Thus, let $(s_0,b_0)\in I$ be such that $s_0$ is the biggest and $b_0$ is the smallest one.
Inductively we relabel elements in $I$ as follows.
Let $(s_{i+1},b_{i+1})\in I$ be such that $s_{i+1}<s_i<b_i<b_{i+1}$ and subject to that condition $s_{i+1}$ is the biggest and $b_{i+1}$ is the smallest one. By  Observation~\ref{obs:interleave}
all the elements in $I$ can be ordered as follows \begin{equation}\label{eqn:order}{\bf (s_0,b_0)},
(s_{0,1},b_0),\ldots, (s_{0,i_0},b_0), (s_0,b_{0,1}),\ldots (s_0,b_{0,j_0}), {\bf (s_1,b_1)}, (s_{1,1},b_{1}) \ldots, (s_{1},b_{1,j_1}),  {\bf (s_2,b_2)}, \ldots\end{equation}
where $s_{k,i+1}<s_{k,i}<s_k$ and $b_{k,i+1}>b_{k,i}>b_k$.
For example, the ordering corresponding to the graph in Figure~\ref{fig:goingUp} is
$(4,6),(3,6),(2,6),(4,7),(3,7),(2,7)$.
Let $E(s,b)$ and $E'(s,b)$, respectively, denote the set of all the edges incident to $v$
contained in an $s$-cap and $b$-cup, where $(s,b)\in I$. Thus, $E(s,b) \cup E'(s,b)$ contain edges incident to $v$
 contained in an interleaving pair that yields a restriction on rotations at $v$ witnessed by $(s,b)$.
Note that $E(s,b) \cap E'(s,b) = \emptyset$.
By Observation~\ref{obs:growth}, $E(s_{k,j+1},b_k)\subseteq E(s_{k,j},b_k)$ and $E'(s_k,b_{k,j+1})\subseteq E'(s_k,b_{k,j})$.
 The restrictions witnessed by $(s,b)$
correspond to the  following condition. In the rotation at $v$ the edges in $E(s,b)$ follow the edges in $E'(s,b)$.
Indeed, otherwise we have a four-tuple of edges $e_1,e_2,f_1$ and $f_2$ incident to $v$,
 such that $e_1,f_1\in P_1$ and $e_2,f_2\in P_2$, where $P_1$ and $P_2$ form an interleaving pair of an $s_i$-cap and $b_i$-cup,
violating the restriction $\{e_1f_1\}\{e_2f_2\}$ on the rotation at $v$.
However, such a four-tuple is not possible in an embedding by Theorem~\ref{thm:characterization}.

Let $L_i=E(s,b)$ and $R_i=E'(s,b)$, respectively, for $(s,b)\in I$, where $i$ is the
index of the position of $(s,b)$ in~$(\ref{eqn:order})$.
Note that $E(s_{i+1},b_{i+1})\cup E'(s_{i+1},b_{i+1}) \subseteq E(s_{i},b_{i})\cup E'(s_{i},b_{i})$.
Our intermediate goal of reducing our problem to the circular-ones property testing would be easy to accomplish if $I$ consisted only of intervals of the form $(s_i,b_i)$ defined above.
However, in $I$ there might be intervals of the form $(s_i,b)$, $b\not=b_i$, or $(s,b_i)$, $s\not=s_i$.
Hence, we cannot just put $L_i':=L_i$ and $R_i':=R_i$ for all $i$, since
we do necessarily have the condition $L_{i+1} \cup R_{i+1} \subseteq L_i \cup R_i$ satisfied for all $i$.


\paragraph{Definition of $\mathcal{S}'$.}
Let $\mathcal{S}=\{L_i,R_i|\ i=1,\ldots \}$. We obtain $\mathcal{S}'$ from $\mathcal{S}$ by deleting the least
number of elements from $L_i$'s and $R_i$'s so that
$L_{i+1}' \cup R_{i+1}' \subseteq L_i' \cup R_i'$ for every $i$.
More formally, $\mathcal{S}'$ is defined recursively as $\mathcal{S}_m'$, where
$\mathcal{S}_1'=\{L_1',R_1'| \ L_1'=L_1, \ R_1'=R_1 \}$ and $\mathcal{S}_{j}'=\mathcal{S}_{j-1}\cup \{L_j',R_j'| \ L_j'=L_j\cap (L_{j-1}'\cup R_{j-1}'),
 R_j'=R_j\cap (L_{j-1}'\cup R_{j-1}')\} $.
 Luckily, the following lemma lying at the heart of the proof of our result shows that information
contained in $\mathcal{S}'$ is all we need.

\bigskip
\begin{figure}[htp]

\centering
\includegraphics[scale=0.7]{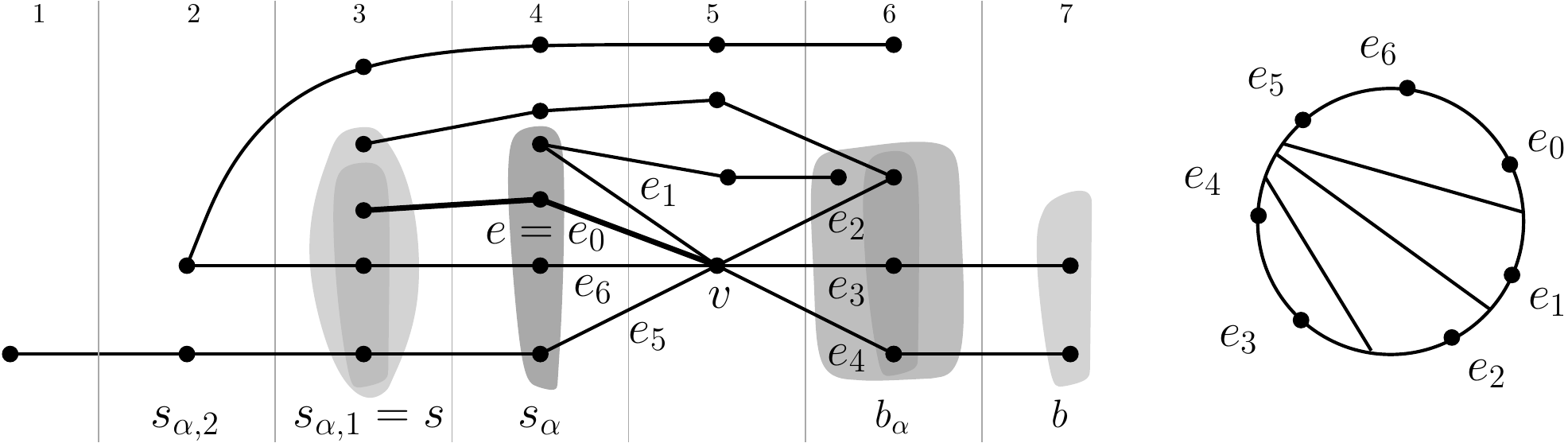}
\caption{A  subdivided star (on the left) with the center $v$, and some restrictions
  on the set of rotation at $v$ (on the right) corresponding to the intervals $(s_\alpha, b_\alpha), (s_{\alpha,1},b_\alpha)$ and $(s,b)$.
  We have $\{e_0,e_5,e_6\}=E(s_{\alpha,1},b_\alpha)\subseteq E(s,b)=\{e_0,e_2,e_5,e_6\}$ and $\{e_3,e_4\}=E'(s,b)\subseteq E'(s_{\alpha,1},b_\alpha)=\{e_3,e_4,e_1,e_2\}$.
Thus, by removing $e_0$ from $E(s,b)$ we obtain the same restrictions on the rotation at $v$.}
\label{fig:goingUp}
\end{figure}

\begin{lemma}
\label{lemma:star}
We can cyclically order the elements in $S$ so that every pair $L_i',R_i'$ in $\mathcal{S}'$ gives rise to two disjoint cyclic intervals
if and only if $(G,\gamma)$ admits an $x$-bounded embedding.
\end{lemma}

\begin{proof}
The proof of the lemma is by a double-induction. In the ``outer--loop'' we induct over $|\mathcal{S}'|/2$ while respecting the order of pairs
 $L_{i},R_{i}$ given by~(\ref{eqn:order}). In the ``inner--loop'' we induct over the size of $S$, where in the base
 case of the $j^\mathrm{th}$ step of the ``outer--loop''
  we have $S_{j,0}=L_{j}' \cup R_{j}'$. In each ${k}^\mathrm{th}$ step of the ``inner--loop''
 we assume by induction hypothesis that a cyclic ordering $\mathcal{O}$ of $S$
 satisfies all the restrictions imposed by $\{L_{i}, R_{i}| i=1,\ldots, j-1\}$ and
 $L_j \cap S_{j,k-1}, R_j \cap S_{j,k-1}$. Clearly, once we show that $\mathcal{O}$ satisfies restrictions imposed by $L_j \cap S_{j,k}, R_j \cap S_{j,k}$,
 where $S_{j,k}=S_{j,k} \cup \{e\}$ and  $e\in (L_{j} \cup R_{j})\setminus S_{j,k-1}$ we are done.

Refer to Figure~\ref{fig:goingUp}.
Roughly speaking, by~(\ref{eqn:order}) a  ``problematic'' edge $e$ is an initial edge on a path starting at $v$ that never visits a cluster $b_\alpha$ after passing through the cluster $s_{\alpha}$ such that $e\in E(s_\alpha, b_\alpha)$ (or vice-versa with $E'(s_\alpha, b_\alpha)$).
 The edge $e$ is an \emph{$(\alpha,\beta)$-lower trim} (or $(\alpha,\beta)$-\ding{33}) if the lowest index $i$ for which $e\not\in L_i' \cup R_i'$ corresponds to $E(s_{\alpha,\beta},b_\alpha)\cup E'(s_{\alpha,\beta},b_\alpha)$, where $\beta >0$.
 Analogously, the edge $e$ is an \emph{$(\alpha,\beta)$-upper trim} (or $(\alpha,\beta)$-\ding{35}) if the lowest index $i$ for which $e\not\in L_i' \cup R_i'$ corresponds to $E(s_\alpha, b_{\alpha,\beta})\cup E'(s_\alpha, b_{\alpha,\beta})$, where $\beta >0$.
By~(\ref{eqn:order}) and symmetry (reversing the order of clusters) we can assume that $e$ is an $(\alpha, \beta)$-\ding{33}, and $e\in E(s_{\alpha,\beta-\beta'},b_\alpha)$, for some $\beta'>0$,
 where $s_{\alpha,0}=s_\alpha$, and $e\in E(s,b)=L_j$, where $s=s_{\alpha,\beta-\beta'}$ and $b>b_\alpha$, following $E(s_{\alpha,\beta},b_\alpha)$ in our order.
Moreover, we pick $e$ so that $e$ maximizes $i$ for which $e\in L_i' \cup R_i'$. We say that $e$ was ``trimmed'' at the $(i+1)^{\mathrm{th}}$ step.

\begin{figure}[htp]

\centering
\includegraphics[scale=0.7]{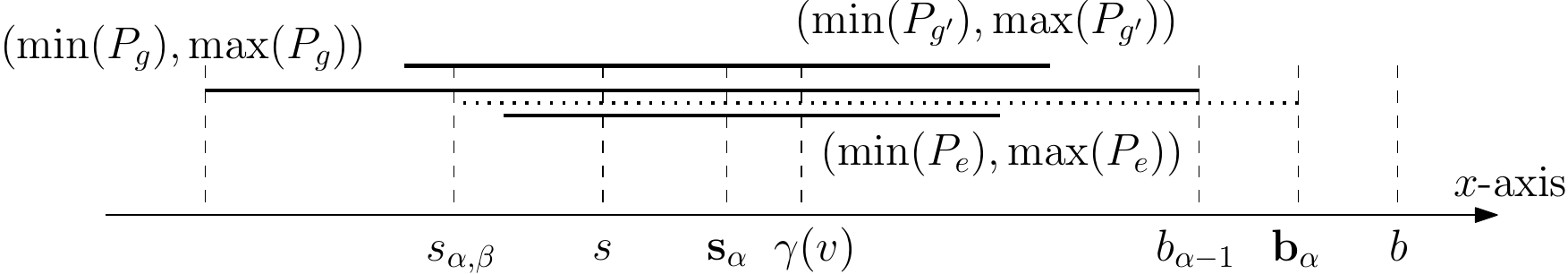}
\caption{Three intervals of clusters corresponding to three paths that start at $v$: $P_e$ that  passes through $e$ and ends in the first vertex in
 the cluster  $s_{\alpha,\beta-1}$, $P_{g'}$ that passes through $g'$ and ends in a leaf, and $P_{g}$ that ends in the first
vertex of the cluster $s''$. (An alternative interval for $P_{g}$ is dotted.) Here, $g'$ was ``trimmed'' before $e$.}
\label{fig:xAxis1}
\end{figure}

 Thus, $e$ is contained in $E(s,b)$ for some $s,b$ such that $E(s,b), E'(s,b)$
follows \\ $E(s_{\alpha,\beta},b_\alpha), E'(s_{\alpha,\beta},b_\alpha)$ in our order. However, it must be that
\begin{equation}
\label{eqn:subset1}
E(s_{\alpha,\beta-\beta'}=s,b_\alpha)\subseteq E(s,b) \  \mathrm{and} \  E'(s,b)\subseteq E'(s_{\alpha,\beta-\beta'}=s,b_\alpha),
\end{equation}
where the first relation follows directly from the fact $b>b_\alpha$ and the second relation is a direct consequence of Observation~\ref{obs:growth}.
In what follows we show that~(\ref{eqn:subset1}) implies that $\mathcal{O}$ satisfies all the required restrictions involving $e$.
We consider an arbitrary four-tuples of edges $e_1',e_2',e_3' \in S_{j,k-1}$ that together with $e$
gives rise to a restriction $\{e_1'e_2'\}\{e_3'e\}$ on $\mathcal{O}$ witnessed by $(s,b)$.
 The incriminating four-tuple must also contain
an element from $E(s,b)\setminus E(s,b_\alpha)$, let us denote it by $f=e_3'$.
Indeed, otherwise by~(\ref{eqn:subset1}) the restriction is witnessed by~$(s,b_\alpha)$ and we are done by induction hypothesis.
 Then $e_1',e_2'\in E'(s,b)$.
For the sake of contradiction we suppose that the order $\mathcal{O}$  violates the restriction $\{e_1'e_2'\}\{ef\}$.
Let $g\in L_{i'}'\subseteq E(s_{\alpha,\beta},b_{\alpha})$, for some $i'$. Note that $g$ exists (see Figure~\ref{fig:xAxis1}) for
if an edge $g' \in  E(s_{\alpha,\beta},b_{\alpha})$ is not in $L_{i'}'$ it means that $g'$ was ``trimmed'' before
$e$ and we can put $g$ to be an arbitrary element from $E(s'',b'')$ minimizing $s''$ appearing before
$E(s_{\alpha},b_{\alpha})$ in our order.

Here, the reasoning goes as follows.
Let $P_{g'}$ denote the path from $v$ passing through $g'$ and ending in a leaf.
Recall that $s_i$'s are decreasing and $b_i$'s are increasing
as $i$ increases. Thus, if we ``trimmed'' $g'$ before $e$, it had to be a \ding{33} by $s_{\alpha,\beta}<s_{\alpha}$, but then
there exists a path starting at $v$ that reaches a cluster with a smaller index than is reached by $P_{g'}$
before reaching even the cluster $b_{\alpha-1}<b_{\alpha}$.
Note that the edge $g$ can be also chosen as an edge in $E(s_{\alpha,\beta},b_{\alpha})$ minimizing $i$ such that the path starting at $v$
passing through $g$ has a vertex in the $i^\mathrm{th}$ cluster. This choice of $g$ plays a crucial role in our proof of the extension of the lemma
for trees.

 Thus, $g\in S_{j,k-1}$ by the choice of $e$, since $e\not\in L_{i'}'$.
Note that $g\in E(s,b_\alpha)$, and hence, $g\in E(s,b)$ by~(\ref{eqn:subset1}).
By Observation~\ref{obs:alegebera}
a restriction $\{e_1'e_2'\}\{fg\}$ is violated as well due to the restriction
 $\{e_1'e_2'\}\{eg\}$, that $\mathcal{O}$ satisfies by induction hypothesis,  witnessed by $(s,b_\alpha)$.  However, by~(\ref{eqn:subset1})  $\{e_1'e_2'\}\{fg\}$
is  witnessed by $(s,b)$ and we reach a contradiction with induction hypothesis.
\end{proof}

By Lemma~\ref{lemma:star} we successfully reduced our question to the problem stated above.
The  problem slightly generalizes the algorithmic question considered by  Hsu and McConnell~\cite{HC03} about testing  0--1 matrices for circular ones property.
An almost identical problem
of testing 0--1 matrices for consecutive ones property was already considered
by Booth and Lueker~\cite{Booth1976335} in the context of interval and planar graphs' recognition.
A matrix has the \emph{consecutive ones}  property if it admits
a permutation of columns resulting in a matrix in which ones  are consecutive in every row.
Our generalization is a special case of the related problem of simultaneous PQ-ordering considered recently by  Bl{\"{a}}sius and Rutter~\cite{BR14}.
In our generalization we allow some elements in the matrices to be ambiguous, i.e., they are allowed to play the roles of both zero or one.
However, we have the property that an ambiguous symbol can have only ambiguous symbols underneath in the same column.

The original algorithm in~\cite{HC03} processes the rows of the 0--1 matrix in an arbitrary order one by one.
In each step the algorithm either outputs that the matrix does not have the circular ones property and stops,
 or produces a data structure called the \emph{PC-tree} that stores \emph{all the
permutations} of its columns witnessing the circular ones property for the matrix consisting of the processed rows.
(The notion of PC-tree is a slight modification of the well-known notion of PQ-tree.)
The columns of the matrix corresponding to the elements of $S$ are in a one-to-one correspondence with the leaves of the PC-tree,
 and a PC-tree produced at every step
is obtained by a modification of the PC-tree produced in the previous step.
Let $\mathcal{Q}_i$ denote the set of permutations captured by the PC-tree after we process the first $i$ rows of the matrix.
Note that $\mathcal{Q}_{i+1}\subseteq \mathcal{Q}_i$.
By deleting some leaves from a PC-tree $T$ along with its adjacent edges we get a PC-tree $T'$
such that $T'$ captures exactly the permutations captured by $T$ restricted to their undeleted leaves.

The original algorithm in~\cite{HC03} runs in a linear time (in the number of elements of the matrix)  The straightforward cubic running time of our algorithm can be improved to a quadratic one.

\paragraph{Running time analysis.}
Let $l$ denote the degree of the center $v$ of the star $G$. 
Let $l_1<\ldots < l_{k-1}$ denote the lengths of paths ending in leaves in $G$ starting at $v$.
Thus, for each $l_i$ there exists such a path of length $l_i$ in $G$.
Let $l_i'$ denote the number of such paths starting at $v$ of length $l_i$.
The number of vertices of $G$ is $n=1+\sum_{i=0}^{k-1}l_il_i'$.
Let $l=\sum_{i=0}^{k-1}l_i$. Let $l'=\sum_{i=0}^{k-1}l_i'$.

Note that a path of length at most $l_i$ cannot ``visit'' more than $l_i$ clusters.
Thus, the number of 0's and 1's in the matrix corresponding to $(G,T)$ is upper bounded
by  $O\left(\sum_{i=0}^{k-1}\left(l'-\sum_{j=0}^{i-1}l_i'\right)l_i^2\right)$.
Indeed, each row of the matrix correponds to a pair of clusters 
and we have  $l'-\sum_{j=0}^{i-1}l_i'$ paths of length at least $l_i$.

In order to obtain a quadratic (in $n$) running time we need to 
upper bound the previous expression by $(\sum_{i=0}^{k-1}l_il_i')^2<n^2$.

We have the following 
 
 $$\sum_{i=0}^{k-1}\left(l'-\sum_{j=0}^{i-1}l_i'\right)l_i^2 \leq l\sum_{i=0}^{k-1}l_il_i' \leq \left(\sum_{i=0}^{k-1}l_il_i'\right)^2$$
 
 where the second inequality is obvious.
To show the first one we proceed as follows.

Consider the region $R$ of the plane bounded by the part of $x$-axis
between $(0,0)$ and $(l',0)$; a vertical line segment
from $(l',0)$ to $(l',l_{k-1})$; and a ``staircase polygonal line'' 
from $(l',l_{k-1})$ to $(0,0)$ with horizontal segments of lengths
$l_{k-1}',l_{k-2}',\ldots, l_0'$ and vertical segments of lenghts
$l_{k-1}-l_{k-2}, l_{k-2}-l_{k-3}, \ldots ,l_1-l_0,l_0$.
Thus, the polygonal line has vertices  \\ $(0,0), (0,l_0),(l_0',l_0), (l_0',l_1), (l_0'+l_1',l_1), \ldots, (l'-l_{k-1}', l_{k-1}), (l', l_{k-1})$.
Let $\mathcal{V}$ denote a three-dimensional set living in the Euclidean three-space obtained 
as a product of $R$ with the interval $[0,l]$ of length $l$ so that
 $\mathcal{V}$ vertically projects to $R$, and we assume that the base $R\times 0$ 
is contained in the $xy$-plane, and the rest of $\mathcal{V}$ is above this plane.
Note that the volume of $\mathcal{V}$ is exactly $ l\sum_{i=0}^{k-1}l_il_i'$.

\begin{figure}[htp]
\centering
\includegraphics[scale=0.7]{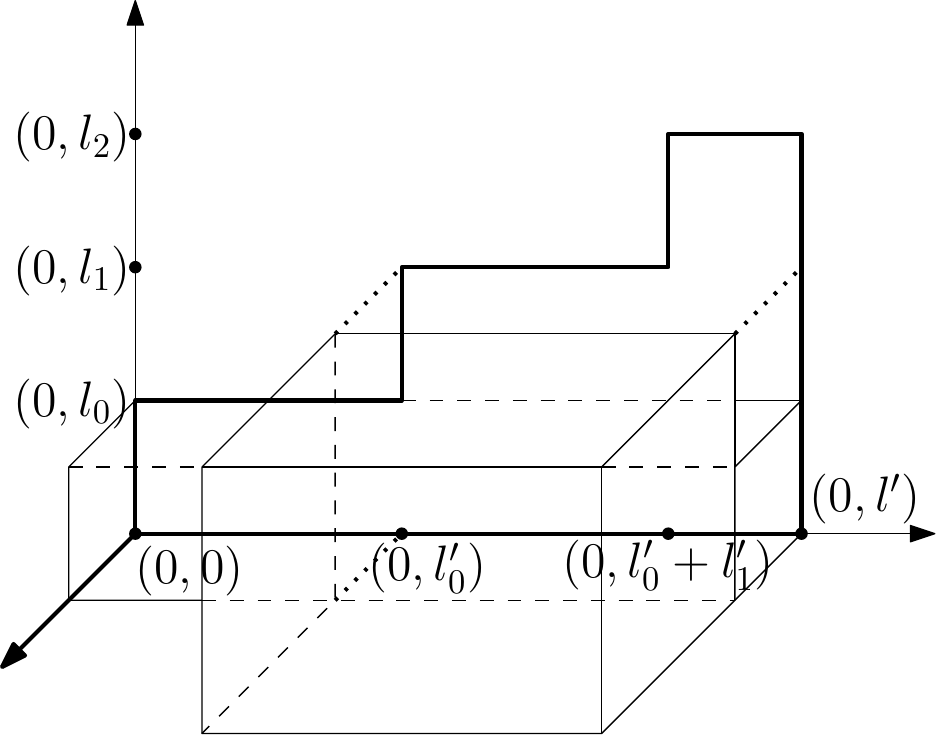}
\caption{The packing of boxes inside $\mathcal{V}$. Only two boxes out of three are shown.}
\label{fig:analysisByBoxes} 
\end{figure}

The expression $\sum_{i=0}^{k-1}\left(l'-\sum_{j=0}^{i-1}l_i'\right)l_i^2$
can be viewed as the sum of volumes of $k$ three-dimensional boxes with integer coordinates. 
Now, it is enough to pack the boxes inside $\mathcal{V}$.
We put the $i^\mathrm{th}$ box with dimensions $\left(l'-\sum_{j=0}^{i-1}l_i'\right)\times l_i \times l_i$ in an axis parallel fashion inside 
$\mathcal{V}$ such that its lexicographically smallest vertex has
coordinates $\left(\sum_{j=0}^{i-1}l_j', 0, \sum_{j=0}^{i-1}l_j\right)$.
It is a routine to check that the boxes are pairwise disjoint and contained in
$\mathcal{V}$ (see Figure~\ref{fig:analysisByBoxes} for an illustration). \\

\subsection{Trees}

In what follows we extend the argument from the previous section  to general trees.
Thus, for the remainder of this section we assume that $(G,\gamma)$ is such that $G$ is a tree.
Let $v$ denote a vertex of $G$ of degree at least three. Refer to Figure~\ref{fig:gvstar}. Let $(G_v, \gamma)$ be such that $G_v$ is a subdivided star
centered at $v$ obtained as follows. For each path $P$ from $v$ to a leaf in $G$ we include to $G_v$ a path $P'$ of the same length,
whose vertex at distance $i$ from $v$ 
has the same $\gamma$ value  as the vertex at distance $i$ from $v$ on $P$.
By slightly abusing notation we denote $\gamma$
also the corresponding function from $V(G_v)$.

\bigskip
\begin{figure}[htp]
\centering
\subfloat[]{\includegraphics[scale=0.7]{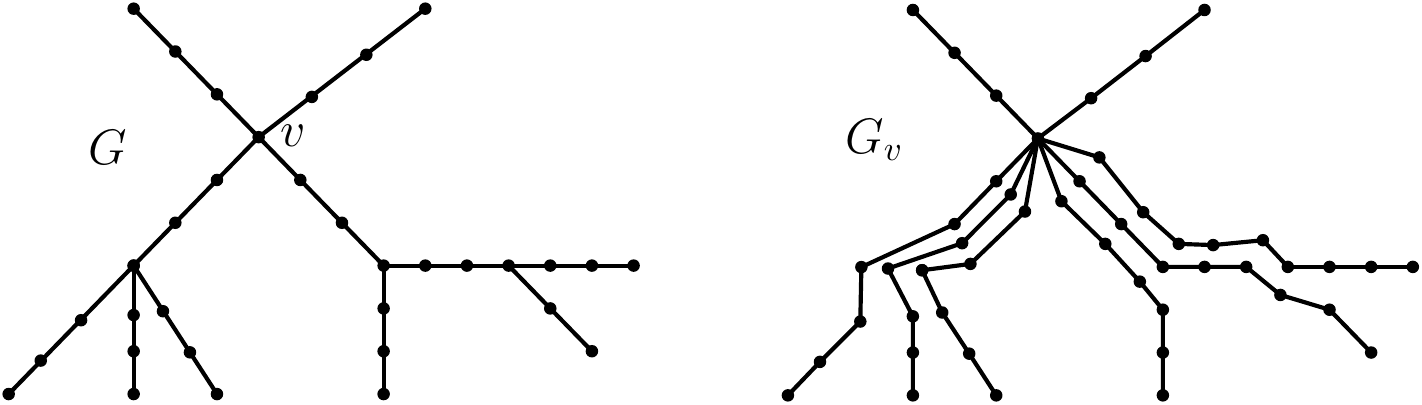}
	\label{fig:gvstar}
	} \hspace{10px}
\subfloat[]{
\includegraphics[scale=0.7]{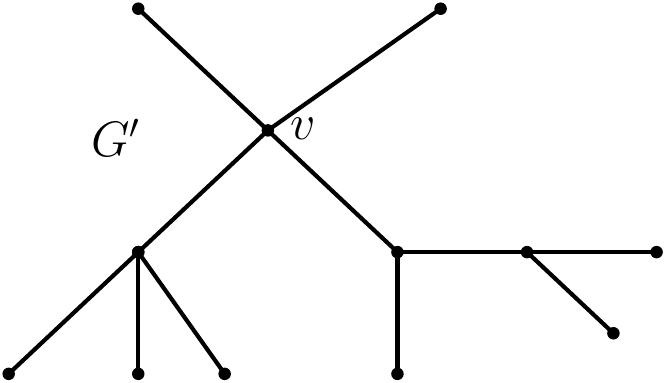}
\label{fig:suppress}}
\caption{(a) A tree $G$ (left) and its subdivided star $G_v$ centered at $v$ (right). (b) The tree $G'$ obtained after suppressing
vertices of degree two in $G$.}
\end{figure}

In the present section we prove Theorem~\ref{thm:AlgTreeXBounded}.

In the light of the characterization from Section~\ref{sec:char} a na\"ive algorithm to test $(G,\gamma)$ for $x$-bounded embeddability 
could use the algorithm from  the previous section to check all   $(G_v, \gamma)$, $v\in V(G)$ with degree at least three,
for $x$-bounded embeddability. However, there are two problems with this approach.
First, we need to take the structure of the tree $G$ into account, since we pass only a limited amount of information about $(G,\gamma)$ to the subdivided stars.
 Second, we need to somehow decide if the common intersection of the sets of possible cyclic orders of leaves of $G$
corresponding to the respective subdivided stars is empty or not.
This would be easy if we did not have ambiguous symbols in our 0--1 matrices corresponding to $(G_v,\gamma)$.

To resolve the first problem is easy, since for each star we can simply
start the algorithm from~\cite{HC03} with the PC-tree isomorphic to $G$, whose all internal vertices are of type $P$ (see~\cite{HC03}
for a description of PC-trees).
This modification corresponds to adding rows into our 0--1 matrix, where each added rows represents the partition of the leaves
of $G$ by a cut edge, or in other words, by a bridge. Let $M_G$ denote the 0--1 matrix representing these rows.
Since we add $M_G$  at the top of the 0--1 matrix with ambiguous symbols corresponding to the given  $(G_v,\gamma)$,
we maintain the property that an ambiguous symbol has only ambiguous symbol underneath.
Moreover, it is enough to modify the matrix only for one subdivided star.
To overcome the second problem we have a work a bit more.

First, we root the tree $G$ at an arbitrary vertex $r$ of degree at least three. Let us suppress all the non-root vertices in $G$ of degree two
and denote by $G'$ the resulting tree (see Figure~\ref{fig:suppress} for an illustration). Let us order $(G_v,\gamma)$'s, where
the degree of $v$ in $G$ is at least three, according to the distance of $v$ from $r$ in $G'$ in a non-increasing manner.
Thus,  $(G_v,\gamma)$ appears in the ordering after all the subdivided stars $(G_u,T)$ for the descendants $u$ of $v$.
For a non-root $v$ we denote by $P_v$ the path in $G$ from $v$ to its parent in $G'$.
Let $I_v=(\min(P_v),\max(P_v))$ denote the interval corresponding to $v$.
Let $I_r=(\min (G), \max (G))$.
Let $M_v$ denote a 0--1 matrix with ambiguous symbols defined by $(G_v,\gamma)$ as in Section~\ref{sec:star}, where each row corresponds
to an interval $(s,b)$ and each column corresponds to an edge incident to $v$ in $G'$ or equivalently to a leaf of $G_v$, and hence, to a leaf of $G$.
In every  0--1 matrix $M_v$ with ambiguous symbols representing $(G_v,\gamma)$ for $v\in G'$ with degree at least three
we delete rows that correspond to intervals $(s,b)$ strictly containing  $I_v$, i.e., $s<\min(P_v)<\max(P_v)<b$.
Let $M_v'$ denote the resulting matrix for every $v$. 

\paragraph{Running time analysis.}
We obtain a cubic running time due to the fact that there exists $O(|V(G)|)$ subdivided
stars $(G_v, T)$, $v\in V(G)$, each of which accounts for $O(|V(G)|^2)$ rows in $M$.

\paragraph{Definition of the matrix $M$.}
Refer to Figure~\ref{fig:matrica}.
Let us combine the obtained matrices $M_G$ and all $M_v'$ for $v\in V(G)$, in the given order so that the rows of $M_v'$ for some $v$ are
added at the bottom of already combined matrices. Let $M'$ denote the resulting matrix.
We replace in $M'$ the minimum number of 0--1 symbols by ambiguous symbols so that the resulting matrix has only
ambiguous symbols below every ambiguous symbol in the same column. Let $M$  denote the resulting matrix.
It remains to show the following lemma.

\begin{lemma}
\label{lemma:treelemma}
The matrix $M$ has circular ones property if and only if $(G,\gamma)$ admits an $x$-bounded embedding.
\end{lemma}

\begin{proof}
We claim that $M$ has circular ones property if and only if $(G,\gamma)$ admits an $x$-bounded embedding.
The ``if'' direction is easy. For the ``only if'' direction we proceed as follows.

A path $P$ in $G$ starting at $v$ is \emph{v-represented} by a column of $M$ or $M_v'$ if the column
corresponds to a leaf $u$ connected with $v$ by a path containing $P$. Note that a path $P$ can be $v$-represented
by more than one column.
A path $P$ in $G$ starting at $v$ is \emph{limited} by the interval $(s,b)$
if $s<\min(P)<\max(P)<b$. Let us assume that $P$ joins $v$ with a leaf.
Note that the column of $M_v'$ $v$-representing $P$ limited by $(s,b)$
 contains the ambiguous symbol in the row corresponding to $(s,b)$.
We consider an interleaving pair of an $s$-cap $P_1$ and $b$-cup $P_2$ that are not disjoint.
Since $G$ is a tree, $P_1$ and $P_2$ share a sub-path $P'$
(that could degenerate to a single vertex). Let $v'$ denote the vertex of $P'$ closest to the root $r$.
If the interval $(s,b)$  does not strictly contain $I_{v'}$ we let $v:=v'$.
If the interval $(s,b)$  strictly contains $I_v$ we let $v$ be the closest ancestor of $v'$ in $G'$ for which
 $(s,b)$ does not strictly contain $I_{v}$. Note that at least the root would do.
 Note that it is possible that $v$ belongs to $P'=P_1 \cap P_2$, that it does not belong to $P_1\cup P_2$, and that it belongs
 to exactly one of $P_1$ and $P_2$. However, by the definition of $v$, if $v$ does not belong to $P_1 \cup P_2$ then none of the paths $P_1$ and $P_2$ can be extended into a path containing $v$.  See Figure~\ref{fig:cases}.

  This property is crucial, and it implies that a row of $M_v'$ corresponding to $(G_v,\gamma)$ gives rise to the restriction on the order
 of leaves corresponding to the pair $P_1$ and $P_2$. This in turn implies that in $M'$ there
 exists a row having
ones in two columns $v$-representing $P_{11}$ and $P_{12}$,
where $P_{i1}$ and $P_{i2}$ denote the paths in $G$ joining $v$ with the end vertices of $P_i$ for $i=1,2$,
and zeros in two columns $v$-representing $P_{21}$ and $P_{22}$ (or vice versa),
However, we need to show that there exists such a row in $M$, or similarly as in the case of subdivided stars, that the corresponding restriction on the order of leaves of $G$ is implied by
other rows, if such a row does not exist.

 \bigskip
\begin{figure}[htp]
\centering
\includegraphics[scale=0.7]{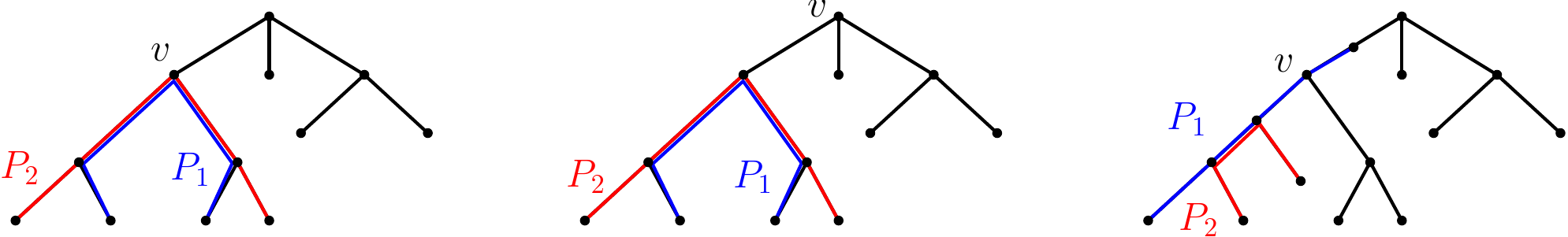}
\caption{Three distinct placements of $v$. On the left, $v$ belongs to $P_1 \cap P_2$, in the middle, $v$ does not belong to $P_1 \cup P_2$, and
on the right, $v$ belongs to exactly one of $P_1$ and $P_2$.}
\label{fig:cases}
\end{figure}

\bigskip
\begin{figure}[htp]
\centering
{
\includegraphics[scale=0.7]{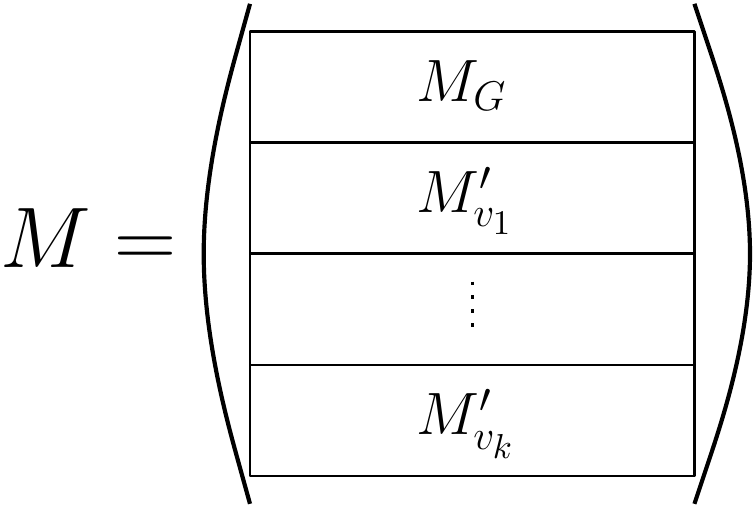}
}
\hspace{10px}
{
\includegraphics[scale=0.7]{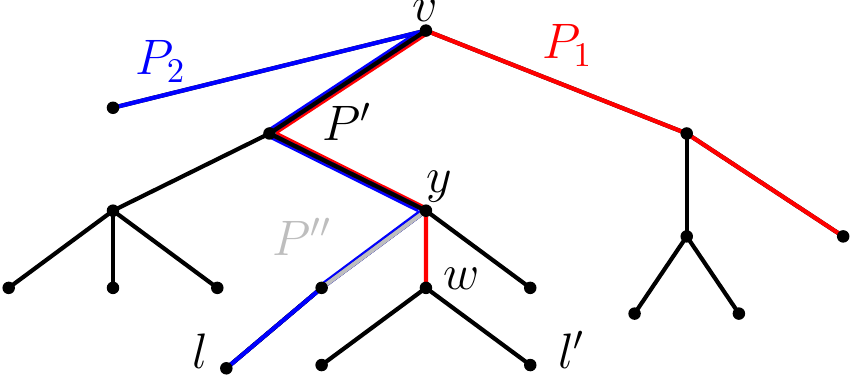}
}
\caption{(a) The composition of the matrix $M$. (b) An interleaving pair of an $s$-cap and $b$-cap $P_1$ and $P_2$ corresponding to
a restriction on a desired cyclic ordering of the leaves of $G$ involving a leaf $l'$ that was ``trimmed'' while processing $(G_y,\gamma)$.}
\label{fig:matrica}
\end{figure}

The PC-tree algorithm processes $M$ from the top row by row.
Enforcing an ambiguous symbol below every ambiguous symbol in $M_y'$ corresponds to ``trimming'' $G$ by shortening every
path $P$ joining $y$ with a leaf $l$ in $G_y$
 limited by the interval corresponding to the currently processed row so that $P$ ends in the parent of $l$ in $G'$.
  Note that we never ``trim'' the path starting at $y$ going towards its parent in $G'$, when
processing $(G_y,\gamma)$, since such a path can be limited only by an interval strictly containing $I_y$.
Also whenever we ``trim'' a path starting at $y$ we keep at least four paths starting at $y$, and thus, at least three paths
going from $y$ towards leaves.
Thus, if $w$, the other end vertex of $P_{\alpha\beta}$ than $v$, is not a descendant of $y$,
there exists at least one column $c$ of $M_y'$ such that $c$ $v$-represents the path $P_{\alpha\beta}$ and $c$ does not contain an ambiguous symbol in $M_y'$.
Unfortunately, if $w$ is a descendant of $y\not=v$, we could ``trim'' all the leaves that are descendants of $w$.
In this case we would like to argue similarly as in the case of a subdivided star that by introducing an ambiguous symbol in $M$  in the row corresponding
to $(s,b)$ of $M_v'$ in  all the columns $v$-representing $P_{\alpha\beta}$,
 we do not disregard required restrictions imposed on the order of leaves of $G$ by our characterization.

Let $y$ be a descendant of $v$ or equal to $v$.
Suppose that  we ``trimmed'' a descendant $l'$ of $w$ or $w$ (if $w$ is a leaf) that is a descendant of $y$, while processing $(G_y,\gamma)$ that appears, of course, before $(G_v,\gamma)$ in
our order or equals $(G_y,\gamma)$. 
Let $L_z(s'',b'')$ and $L_z'(s'',b'')$, respectively, be defined analogously as $E(s'',b'')$ and $E'(s'',b'')$ in Section~\ref{sec:star}
for some $(G_z,T)$ with leaves of $G$ playing the role of the edges incident to the center of the star and some $s''$ and $b''$, $s''<b''$.
We also have an order corresponding to~(\ref{eqn:order}) for every $(G_z,T)$.
By reversing the order of clusters, without loss of generality we can assume that $w$ is an end vertex of $P_1$ which is an $s$-cap.
Refer to Figure~\ref{fig:xAxis23}(a).
If $\gamma(w)\in I_y$ then $w$ cannot be a descendant of $y$, since we have  $I_y\subseteq (\min (P_1\setminus w), \max(P_1 \setminus w))$ and $w\not\in (\min (P_1\setminus w), \max(P_1 \setminus w))$.
 Thus, $\gamma(w)\not\in I_y$, and hence, $s\not\in I_y$. Moreover, $s<\min(I_y)$, since $P_1$ is an $s$-cap.
Since the descendant $l'$ of $w$ was ``trimmed'' while processing $(G_y,\gamma)$ by~(\ref{eqn:order}) there exists
 $L_y(s',b')$ for some $s',b'$ not containing $l'$.
Since the interval $(s',b')$ does not strictly contain $I_y$, contains $s$, and $s<\min(I_y)$ we have $s'\not\in I_y$, $s'<s$ and $b'\in I_y$.
  Moreover, if $y\not= v$ for some row of $M_y'$ we have the corresponding $L_y(s,b')$ containing $l'$
   since there exist at least two paths in $G_y$ from $y$ whose initial pieces correspond to the path from $y$ toward its parent as $v$ has degree at least three. Note that $L_y(s,b')$ can contain only descendants of $y$.
We claim the following (the proof is postponed until later)
\begin{equation}
{L_y(s,b')\subseteq L_v(s,b) \ \mathrm{and} \
L_y'(s,b')\supseteq L_v'(s,b)
\label{eqn:ahaha}}
\end{equation}
  Now, by using~(\ref{eqn:ahaha}) we can extend  the double-induction argument from Lemma~\ref{lemma:star}.

In the same manner as in the previous section $R_i$ and $L_i$ correspond to the $i^\mathrm{th}$ row of $M'$.
We define $\mathcal{S}'$ recursively as $\mathcal{S}_m'$, $m$ is the number of rows of $M'$, where
$\mathcal{S}_1'=\{L_1',R_1'| \ L_1'=L_1, \ R_1'=R_1 \}$ and $\mathcal{S}_{j}'=\mathcal{S}_{j-1}\cup \{L_j',R_j'| \ L_j'=L_j\cap (L_{j-1}'\cup R_{j-1}'),
 R_j'=R_j\cap (L_{j-1}'\cup R_{j-1}')\} $. Let $S_{j,0}=L_j' \cup R_j'$.
We need, in fact, to apply the condition~(\ref{eqn:ahaha}) only when a new leaf $l'$, such that $S_{j,k}=S_{j,k} \cup \{l'\}$,
 added to $S_{j,k-1}$ (playing the role of the edge $e$ from the proof of Lemma~\ref{lemma:star})
 is the only descendant of $w$ in ${S}_{j,k}$.
For the other descendants of $w$, $M_G$ forces the corresponding restriction.
  More formally, we need to show that a cyclic ordering of leaves $\mathcal{O}$ respecting restrictions imposed by the first $j-1$ rows, and
  the columns corresponding to $S_{j,k-1}$ in the $j^\mathrm{th}$ row, respects also restrictions imposed by
  the columns corresponding to $S_{j,k}$ in the $j^\mathrm{th}$ row. Let $\{l'l_1\}\{l_2l_3\}$ be such a restriction for a leaf $l'$ trimmed the most
  recently similarly as for $e$ in the case of  subdivided stars.
  Let the restriction  $\{l'l_1\}\{l_2l_3\}$, $l_1\in L_v(s,b)\setminus L_y(s,b')$, $l'\in L_v(s,b)$  and $l_2,l_3\in L_v'(s,b)$, induced by $S_{j,k}$ in the $j^\mathrm{th}$ row
 correspond to pair of an $s$-cap $P_1$ and $b$-cup $P_2$,
  where the leaf $l'$ $v$-represents a sub-path $P_{\alpha\beta}$ (for $v$ and $P_{\alpha\beta}$ defined as above)
  of $P_1$ ending in $w$ such that $\gamma(w)=s$.

\bigskip
\begin{figure}[htp]
\centering
{
\includegraphics[scale=0.7]{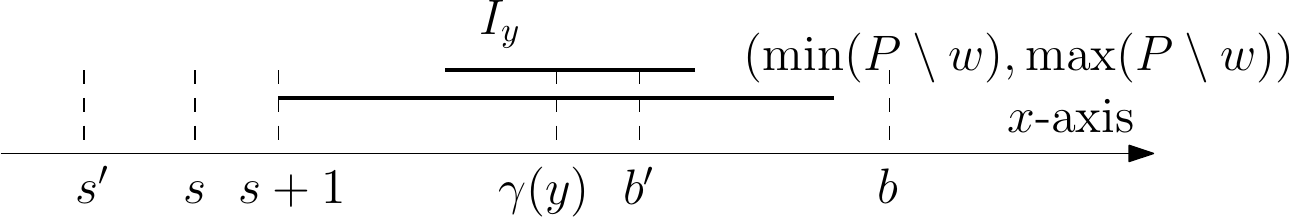}
}
\hspace{5px}
{
\includegraphics[scale=0.7]{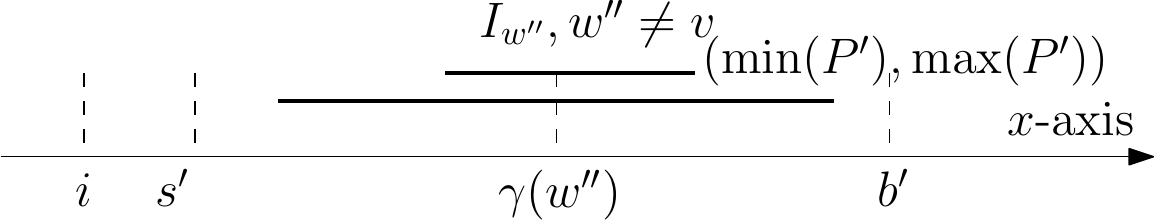}
}
\caption{(a) The interval corresponding to $P\setminus \{w\}$ containing $I_y$. (b) ) The interval corresponding to $P$ containing $I_{w'}$.}
\label{fig:xAxis23}
\end{figure}

  First, we assume that $y=v$. First, note that $l_2,l_3\in E_v'(s,b')$ by~(\ref{eqn:subset1}).
  We proceed by the same argument as in Lemma~\ref{lemma:star}, since we have~(\ref{eqn:order})
  for $G_v$. Here, $s,s',b$ and $b'$, respectively, plays the role of $s, s_{\alpha,\beta},b$ and $b_\alpha$.
After we find $l''\in  L_v(s,b')$ that was ``trimmed'' after $l'$ by induction hypothesis we have
   $\{l'l''\}\{l_2l_3\}$ and $\{l''l_1\}\{l_2l_3\}$. Thus, by Observation~\ref{obs:alegebera} we are done.
The only problem could be that we cannot find a leaf $l''\in L_v(s,b')$ (analogous to the edge $g$) in the proof of Lemma~\ref{lemma:star}
that was not ``trimmed'' before $l'$
   since all such leaves could be potentially ``trimmed'' while processing previous $(G_{y'},T)$. However, recall that we ``trim''
   only descendant of such $y'$ while processing $(G_{y'},T)$.  We show that we can pick $l''$ such that this does not happen.

Refer to Figure~\ref{fig:xAxis23}(b).
    Indeed, consider the path from $v$ to its descendant $w'$ that is an ancestor of $l''$
   such that $w'$ is an end vertex of an $s'$-cap witnessing presence of $l''$ in $L_v(s',b')$.
 Among all possible choices of $l''$ and $w'$, where, of course $l''$ is
   in the subtree rooted at $w'$, let us choose the one minimizing $i$ such that the subtree rooted at $w'$ contains a vertex in the $i^\mathrm{th}$  cluster.
   Moreover, we assume that $l''$ can be reached from $w'$ by following a path passing through the vertex in the $i^{\mathrm{th}}$ cluster.
   Let $P'$ denote the path between $v$ and $w'$.
   Let $w''$ be an ancestor of $w$  that belongs to $V(G')$ and is not an ancestor of $v$.
(For other choices of $w''$ we cannot trim $l''$ while processing $G_{w''}$.)
  If $v=w''$ we show that we cannot ``trim''  $l''$ while processing $G_{w''}$ by the argument from Lemma~\ref{lemma:star}.
Otherwise, suppose that $w''\not=v$.
   Since $I_{w''}\subseteq (\min(P'),\max(P'))$ and $b_\alpha\not\in (\min(P'),\max(P'))$, we cannot ``trim'' the descendant $l''$ of $w'$ while processing $(G_{w''},T)$. For otherwise  $l''$  is a \ding{33}
   and we would get into a contradiction with the choice of $w'$, since there exists a descendant $l'''$ of $w''$
   in a cluster with the index smaller than $i$ good for us.
    The choice of $w''$ and $l'''$ is good since the interval corresponding
     to the row of $M$ with the maximal index, in which a column of the leaf $l''$ has 0 or 1, does not strictly contain $I_{w''}$,
 and thus, the path from $v$ towards $l'''$ that ends in the cluster $s'$ never visits $b'$ cluster.

  Second, we assume that $y\not=v$. In this case we also have that $l_2,l_3\in E_y'(s,b')$ by~(\ref{eqn:ahaha}).
  We again repeat the argument from
Lemma~\ref{lemma:star} in the same manner as for the case  $y=v$.
  Here, again $s,s',b$ and $b'$, respectively, plays the role of $s, s_{\alpha,\beta},b$ and $b_\alpha$.
  Also a leaf $l''\in L_y(s,b')$ that was not ``trimmed'' before $l'$ playing the role of $g$ is found by the analysis in the previous paragraph,
where $y$ plays the role of $v$.

  If $l'$ is not a descendant of $v$, recall that
  there exists at least one ``untrimmed'' leaf $l''$ such that $l''$ $v$-represents the path $P_{\alpha\beta}$.
  Then by induction hypothesis a restriction $\{l''l_1\}\{l_2l_3\}$ witnessed by $(s,b)$ corresponding to the  $j^{\mathrm{th}}$ row gives
  the desired restriction $\{l'l_1\}\{l_2l_3\}$ on $\mathcal{O}$ by Observation~\ref{obs:alegebera}, since we have $\{l''l'\}\{l_2l_3\}$ by $M_G$.

It remains to prove~(\ref{eqn:ahaha}). Refer to Figure~\ref{fig:matrica}(b).
If $v=y$ we are done by the argument in Lemma~\ref{lemma:star}. Thus, we assume that $y\not= v$.
 We start by proving the first relation $L_y(s,b')\subseteq L_v(s,b)$.
If a leaf descendant $l$ of $y$ is in $L_y(s,b')$ then $l\in L_v(s,b)$, since $b>b'>\gamma(y)$ by $b'\in I_y,b\not\in I_y$,
and $s\not\in (\min(P'),\max(P'))$, where $P'$ is a path between $y$ and $v$.
Then the corresponding witnessing path from $y$ towards $l$ of the fact $l\in L_y(s,b)$ can be extended by $P'$ to the
path witnessing $l\in L_v(s,b')$.
In order to prove the second relation $L_y'(s,b')\supseteq L_v'(s,b)$ we first observe that
$L_y'(s,b')$ contains all the leaves that are not descendants of $y$, since $b'\in I_y$.
On the other hand, if a leaf descendant $l$ of $y$ is in $L_v'(s,b)$ then $l\in L_y'(s,b')$, since the corresponding witnessing path
from $v$ toward $l$ of the fact $l\in L_v'(s,b)$ contains a sub-path $P''$ starting at $y$ and ending in the cluster $b'$ due to $b>b'>\gamma(y)$
witnessing $l\in L_y'(s,b')$.

Thus, if the 0--1 matrix $M$ with ambiguous symbols has circular ones property then there exists an embedding of $G$ such that
every interleaving pair is feasible, and hence, by Theorem~\ref{thm:characterization}  graph $(G,\gamma)$ admits an $x$-bounded embedding.
\end{proof}

\section{Theta graphs}

\label{sec:theta}

In this section we extend result from the previous one to the class of   $x$-bounded drawings
whose underlying abstract graph is a \emph{theta graph} defined as a union of internally vertex disjoint paths joining a pair of vertices that we call \emph{poles}. Hence, in the present section $(G,\gamma)$ is such that $G$ is a theta graph. 
Similarly as in Section~\ref{sec:tree} in what follows we assume~$(*)$.

Our efficient algorithm for testing $x$-bounded embeddability of $(G,\gamma)$ relies on the work of 
Bl\"asius and Rutter~\cite{BR14}. We refer the reader unfamiliar with this work
to the paper for necessary definitions. Thus, our goal is to reduce the problem to the 
problem of finding an ordering of a finite set that satisfies constraints 
given by a collection of PC-trees.\footnote{Despite the fact that~\cite{BR14} has the word
``PQ-ordering'' in the title, the authors work, in fact, with un-rooted  PQ-trees, which are our
PC-trees.}

The construction of the corresponding instance $\mathcal{I}$ of the simultaneous PC-ordering for the given $(G,\gamma)$ 
is inspired by~\cite[Section 4.2]{BR14}.
Thus, the instance consists of a star $T_C$ having a P-node in the center  (\emph{consistency tree}),
and a collection of \emph{embedding trees} $T_0,\ldots, T_m$ constructed analogously
as in Section~\ref{sec:tree}.
The DAG (directed acyclic graph) representing $\mathcal{I}$ contains edges $(T_0,T_C,\varphi_1)$, $(T_0,T_C,\varphi_2)$, and $(T_i,T_{i+1},\phi_i)$ for $i=0,\ldots ,m-1$~\footnote{see~\cite[Section 3]{BR14} for
the definition of the DAG representing $\mathcal{I}$.} . Tree $T_0$ (see Figure~\ref{fig:thetaTree}) will consist, besides leaves and their incident edges, only of a pair of $P$-nodes joined by an edge.
It follows that the instance is solvable in a polynomial time by~\cite[Theorem 3.3 and Lemma 3.5]{BR14}\footnote{Using the terminology of~\cite{BR14} the reason is that the instance is 2-fixed.
Therein in the definition of $\mathrm{fixed}(\mu)$ it is assumed that 
every P-node $\mu$ in a  tree $T$ fixes in every parental tree $T_i$ of $T$ at most one 
P-node. This is not true in our instance due to the presence of multi-edges in the DAG.
However, multi-edges are otherwise allowed in the studied model.
 We think that the authors, in fact, meant to say that for every incoming edge of $(T_i,T)$
 the node  $\mu$ fixes in the corresponding projection of $T_i$ to the leaves of $T$ at most one P-node. Nevertheless, we can still fulfill the condition by getting rid of the multi-edge as follows. We introduce two additional copies of $T_0$, let's say $T'$ and $T''$,
 and instead of  $(T_0,T_C,\varphi_1)$ and $(T_0,T_C,\varphi_2)$ we
 put $(T',T_C,\varphi_1)$ and $(T'',T_C,\varphi_2)$. Finally, 
 we put $(T',T_0,\varphi)$ and $(T'',T_0,\varphi)$, where $\varphi$ is the identity.
 It is a routine to check that the resulting instance is 2-fixed.}
The running time follows by~\cite[Theorem 3.2]{BR14}.

\begin{figure}[htp]
\centering
\includegraphics[scale=0.7]{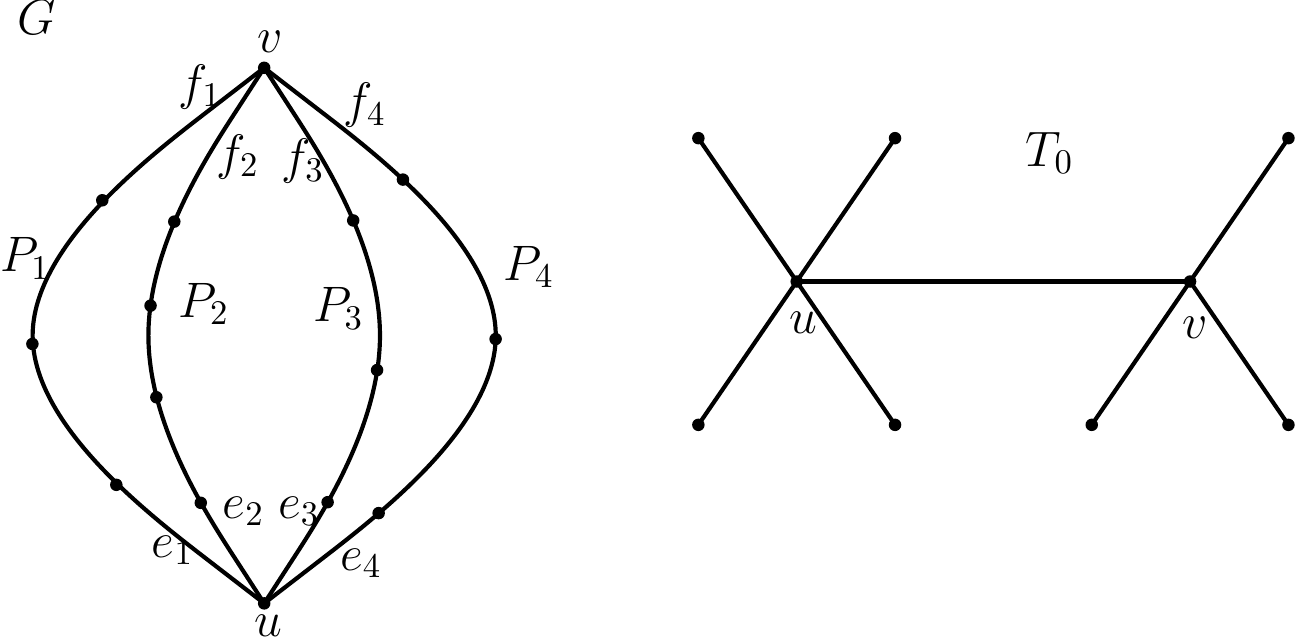}
\caption{The theta-graph $G$ (left) and the PC-tree $T_0$ in
its corresponding $\mathcal{I}$ instance.}
\label{fig:thetaTree}
\end{figure}

\paragraph{Description of $(T_0,T_C,\varphi_1)$, $(T_0,T_C,\varphi_2)$.}
Let $u$ and $v$ denote the poles of $G=(V,E)$.
Let $e_1,\ldots, e_n$ denote the edges incident to $u$.
Let $f_1,\ldots, f_n$ denote the edges incident to $v$.
We assume that $e_i$ and $f_i$ belong to the same  path $P_i$ between $u$ and $v$.
The non-leaf vertices of $T_0$ are $P$-nodes $u$ and $v$, corresponding 
to $u$ and $v$ of $G$, joined by an edge. The $P$-node $u$ is adjacent to $n$ leaves 
and $v$ is adjacent to $n-1$ leaves.
The tree $T_C$ is a PC-tree with a single $P$-node and $n$ leaves.
The map $\varphi_1$ maps injectively every leaf of $T_C$ except one to a leave adjacent to $u$
the remaining leaf is mapped to an arbitrary leave of $v$.
The map $\varphi_2$ maps injectively every leaf of $T_C$ except one to a leave adjacent to $v$
the remaining leaf is mapped arbitrarily such that the map is injective.\\

Let $I_{i}= (\min(P_i), \max(P_i))$ for every $i=1,\ldots,n$.
Let $I_j$ be chosen so that $I_i\subseteq I_j$ implies $i=j$.
Thus, $I_j$ is spanning a minimal number of clusters among $I_i$'s.
Let $I_{\alpha}:=I_j$ and $P_{\alpha}:=P_j$.
The leaves of $\varphi_1$ are mapped to the leaves of $u$ corresponding to edges 
incident to $u$ except for $e_{\alpha}$ and a single leave corresponding to the position of the 
the outer-face,
and we have the analogous compatible correspondence for $\varphi_2$ except that $v$ has 
no leave representing the outer-face that $\varphi_2$ avoids.
Since $T_C$ is the consistency PC-tree, the inherited cyclic orders of end pieces of
paths $P_i$ corresponding to $\varphi_1$ and $\varphi_2$ have opposite orientations. Thus,  one of the arcs  $(T_0,T_C,\varphi_1)$ and $(T_0,T_C,\varphi_2)$ is orientation preserving and the other one
orientation reversing.
\bigskip

\begin{figure}[htp]
\centering
\includegraphics[scale=0.7]{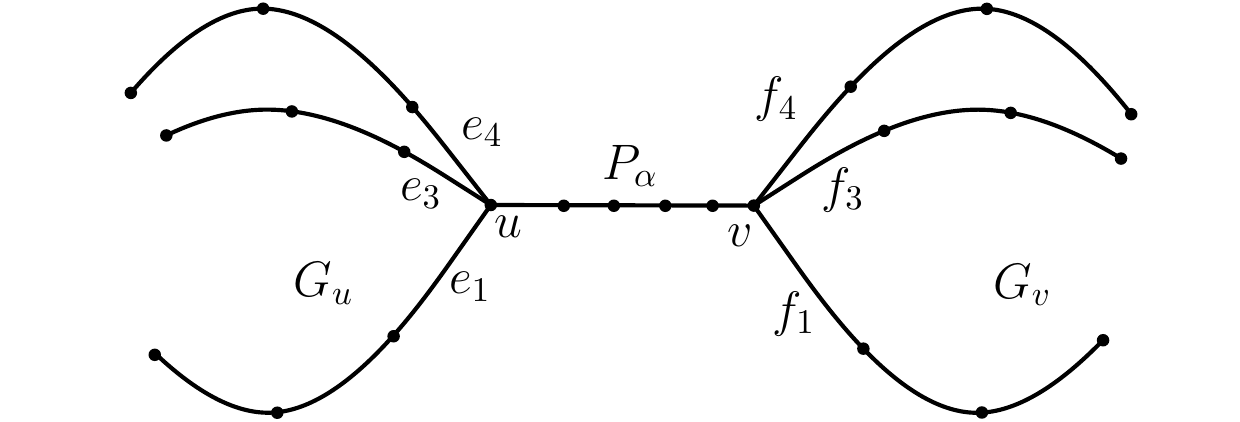}
\caption{The graph $G'$ consisting of $G_u$, $G_v$ and $P_\alpha$.}
\label{fig:thetaTree2}
\end{figure}

\paragraph{Description of $(T_i,T_{i+1},\phi_1)$.}
We will construct $(G',\gamma')$, where $G'$ is a tree,~(see Figure~\ref{fig:thetaTree2}) yielding desired embedding trees $T_i$'s, for $i>0$, in
$\mathcal{I}$. The tree $G'$ is obtained as the union of a pair of vertex disjoint
subdivided star $G_u$ and $G_v$, and the path $P_\alpha$ (defined in the previous paragraph) joining the centers of $G_u$ and $G_v$.
The graph $G_u$ is isomorphic to  $G\setminus (\{v\}\cup E(P))$ 
and $G_v$ is isomorphic to $G\setminus (\{u\}\cup E(P))$.
 The assignment $T'$ of the vertices to clusters
is inherited from $(G,\gamma)$.

The path $T_0T_1\ldots T_m$ in the  DAG of $\mathcal{I}$ corresponds to a variant $M'$ 
of the matrix $M$ from Section~\ref{sec:tree} constructed for $G'$ with additional rows between $M_{G'}'$ and $M_u'$.
 The leaves of $T_1$ corresponds to the columns of $M'$. \\

\noindent 
(i) The tree $T_0$ corresponds to $M_{G'}$. \\ (ii) $T_1$ takes care of the trapped vertices (that
we did not have to deal with in the case of trees).
 \\ (iii) $T_1,\ldots, T_l$, for some $l$, corresponds to $M_u'$, and  $T_{l+1},\ldots, T_m$, to $M_v'$. \\
 
The described trees naturally correspond to constraints on the rotations at $u$ and $v$.
Before we proceed with proving the correctness of the algorithm we describe the last missing
piece of the construction, the PC-tree $T_1$.

\paragraph{Description of $T_1$.}
The PC-tree $T_1$ corresponds to the set of constraints given by the following 0--1 matrix $M''$.
The leaves corresponding to $e_{\alpha}$ are all the leaves incident to $v$, and the leaves
corresponding to $f_{\alpha}$ are the leaves incident to $u$.
The correspondence of the remaining edges incident to $u$ and $v$ to the columns of $M''$ is
given by their correspondence with leaves of $T_1$ explained above.
The matrix $M''$ has only zeros in the column representing the outer-face in the rotation at $u$.
Let us take the maximal subset of edges incident to $u$, whose elements
$e_1,\ldots, e_{n'}$  are ordered (we relabel the edges appropriately) such that
$i<i'$ implies $\min (P_i)<\min(P_{i'})$. For every $j$, $1\leq j \leq n'$, we introduce
a row having zeros in the columns corresponding to $e_1,\ldots ,e_j$ and columns
corresponding to $e_i$ for which there exists $i'$, $1\leq i'\leq j$, $\min (P_i)=\min(P_{i'})$, and having ones
in the remaining columns except the one representing the outer-face.

Let us take the maximal subset of edges incident to $u$, whose elements
$e_1,\ldots, e_{n'}$ are ordered  such that
$i<i'$ implies $\max (P_i)>\max(P_{i'})$. For every $j$, $1\leq j \leq {n'}$, we introduce
a row having zeros in the columns corresponding to $e_1,\ldots ,e_j$ and columns
corresponding to $e_i$ for which there exists $i'$, $1\leq i'\leq j$, $\max (P_i)=\max(P_{i'})$, and having  ones
in the remaining columns except the one representing the outer-face.  

A trapped vertex $v'$ on $P_{i}$ in a cycle $C$ consisting of $P_{i'}$ 
and $P_{i''}$ would violate a constraint $\{e_i'e_i''\}\{e_ie_0\}$,
where $e_0$ is the dummy edge on the outer-face enforced by $T_1$. \\

It remains to prove that the instance $\mathcal{I}$ is a if and only if 
any corresponding witnessing order of the leaves yields an isotopy class of $G$ that
contains an $x$-bounded embedding of $(G,\gamma)$ by Theorem~\ref{thm:characterization}.
This might come as a surprise since some constraints on the rotation system enforced by the original instance $(G,\gamma)$ 
might be missing in $\mathcal{I}$, and on the other hand some additional constraints might be introduced.

\begin{theorem}
\label{thm:theta_alg}
The instance $\mathcal{I}$ is a ``yes'' instance if and only if $(G,\gamma)$ admits an $x$-bounded embedding.
\end{theorem}

\begin{proof}
By the discussion above it remains to prove that the order constraints corresponding to the
trees $T_2,\ldots, T_m$ are all the constraints given by the infeasible
interleaving pairs of paths $P_1'$ and $P_2'$, and possibly additional
constraint enforced by trapped vertices.

First, we note that no additional constraints are introduced due to the fact 
that we might have two copies of a single vertex of $G$ in $G'$. Indeed, such a 
constraint corresponds to a pair of paths $P_1'$ and $P_2'$ intersecting in the copy of $P_{\alpha}$,
such that their union contains at least one whole additional $P_l$, for $l\not=\alpha$.
However, in this case it must be that the two copies of a single vertex in the union of $P_1'$ and $P_2'$ are the endpoints of, let's say $P_1'$. Thus, the constraint of $P_1'$ and $P_2'$ exactly prevents end vertices of $P_2'$ from being trapped
 in the cycle of $G$ obtained by identifying the end vertices of $P_1'$.

\begin{figure}[htp]
\centering
\includegraphics[scale=0.7]{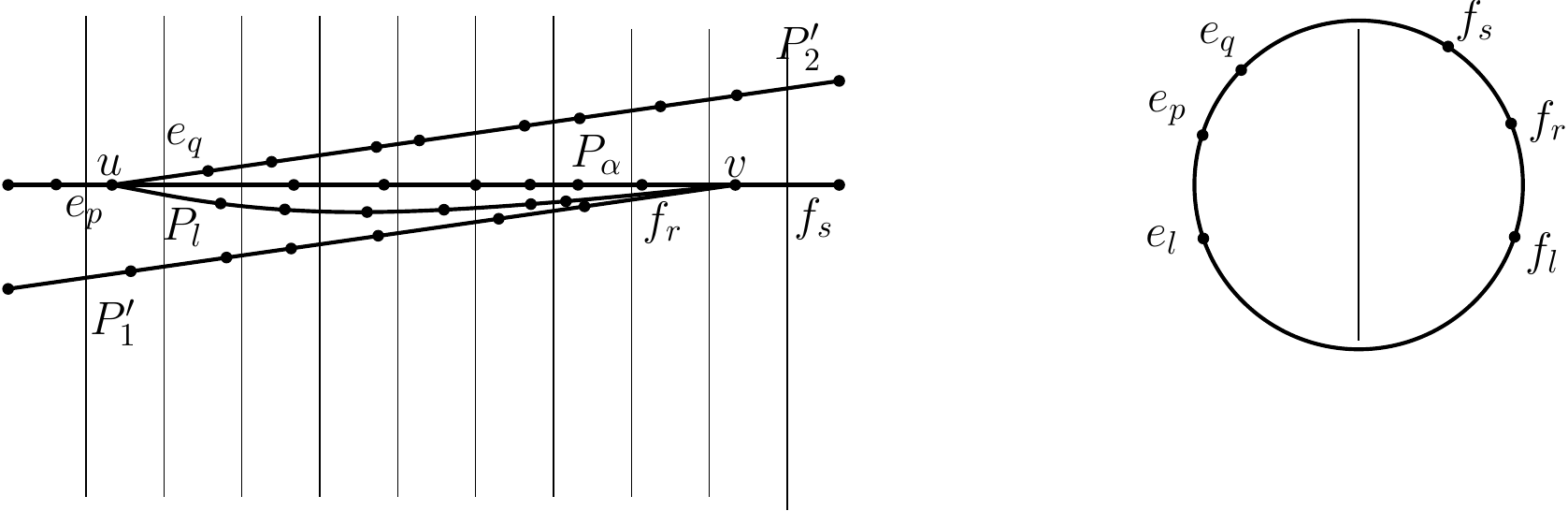}
\caption{The paths $P_1'$ and $P_2'$ meeting in $P_l$ such that $I_{l}= (\min(P_l), \max(P_l))$ 
contains $I_\alpha$ (left). The corresponding cyclic order of the leaves  corresponding to the edges incident to $u$ and $v$ captured by $T_0$ (right).}
\label{fig:thetaTree3}
\end{figure}

  If $P_1'$ and $P_2'$ intersect exactly in the vertex $u$ or $v$ the corresponding constraint is definitely captured. Similarly, if $P_1'$ and $P_2'$ intersects exactly in the path $P_\alpha$.
Otherwise, if $P_1'$ and $P_2'$ intersects in a path $P_l$ 
such that $I_{l}= (\min(P_l), \max(P_l))$ contains $I_\alpha$, we have the corresponding
constraint present implicitly. 

Refer to Figure~\ref{fig:thetaTree3}. Indeed, an ordering $\mathcal{O}$ of the columns of $M'$  witnessing that $\mathcal{I}$ is a ``yes'' instance satisfies 
$\{e_{\alpha}e_l\}\{e_pe_q\}$  by $T_1$, where
$e_p$ and $e_q$, respectively, is the edge incident to $u$ belonging to $P_1'$ and $P_2'$,
and  also $\mathcal{O}$ satisfies $\{f_{\alpha}f_l\}\{f_rf_s\}$, by the consistency tree $T_C$, where
$f_r$ and $f_s$, respectively, is the edge incident to $v$ belonging to $P_1'$ and $P_2'$,
Moreover, by the constraint obtained from the union of
paths  $P_1'$ and $P_2'$ by replacing $P_l$ with $P_{\alpha}$, we obtain
that $\mathcal{O}$ satisfies $\{e_pf_r\}\{e_qf_s\}$. By the constrains of $T_0$ it then 
follows that in the rotation at $u$ 
 the edges $e_l,e_p,e_q$ appear w.r.t. to this order 
 with the opposite orientation as $f_l,f_r,f_s$. Hence, the pair of $P_1'$ and $P_2'$ is feasible with respect to $\mathcal{O}$.

\begin{figure}[htp]
\centering
\includegraphics[scale=0.7]{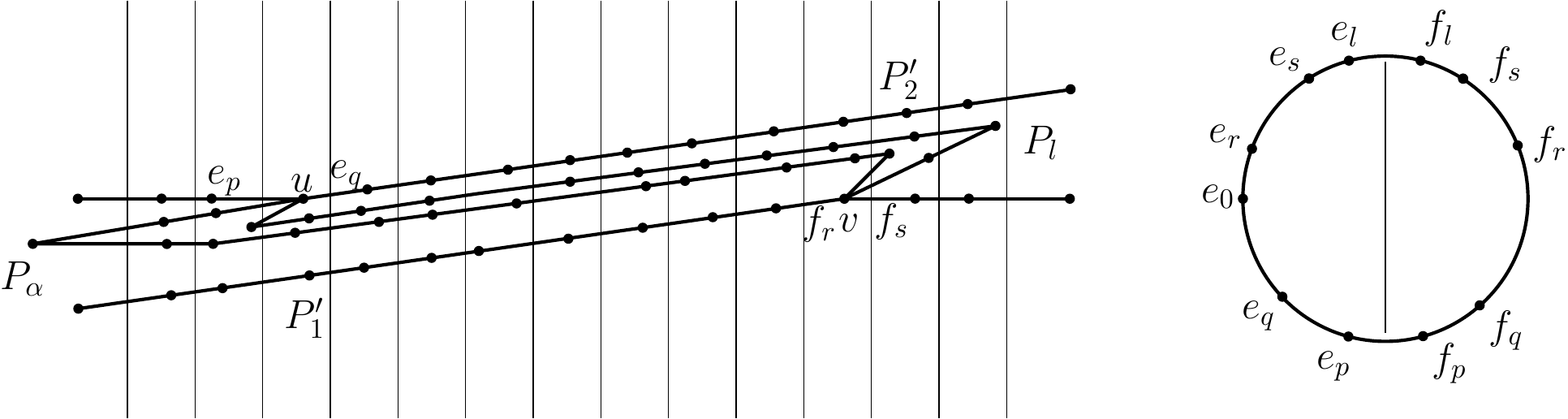}
\caption{The paths $P_1'$ and $P_2'$ meeting in $P_l$ such that $I_{l}= (\min(P_l), \max(P_l))$ 
does not contain $I_\alpha$ (left). The corresponding cyclic order of the leaves  corresponding to the edges incident to $u$ and $v$ captured by $T_0$ (right).}
\label{fig:thetaTree4} 
\end{figure}

Finally, if $I_{l}$ does not contain $I_\alpha$, the previous argument does not apply if the interval \\
$(\min(P_1'\cup P_2'), \max(P_1' \cup P_2'))$ does not contain 
$I_\alpha$. We assume that \\ $$\max(P_1' \cup P_2')>\max(I_l)>\max(I_\alpha)\ge \min(I_l)> \min(P_1' \cup P_2'))\ge \min(I)$$ and handle the remaining cases by the symmetry. We assume that $P_1'$ is a cap
 passing through edges $e_p$ and $f_r$.  We assume that $P_2'$ is a cup
 passing through edges $e_q$ and $f_s$. 
 
 The ordering $\mathcal{O}$ satisfies $\{e_pe_{\alpha}\}\{e_le_q\}$ and $\{f_rf_{\alpha}\}\{f_lf_s\}$.
 By  $T_1$ we have $\{e_qe_s\}\{e_{\alpha}e_l\}$. By $T_0$ and $T_C$,  $e_i$'s and $f_i$'s appear consecutively and they are reverse of each other in $\mathcal{O}$. Thus, we have $\{e_lf_l\}\{e_qf_s\}$.
 
 Refer to Figure~\ref{fig:thetaTree4}.
  A simple case analysis reveals that the observations in the previous paragraph gives us the following.
 If  $\{e_pf_r\}\{e_qf_s\}$ is
  satisfied by $\mathcal{O}$,
  then $\mathcal{O}$ satisfies $\{e_lf_l\}\{e_pe_q\}$ if and only if $\{e_lf_l\}\{f_sf_r\}$.
   On the other hand, if $\{e_pf_r\}\{e_qf_s\}$ is not satisfied by $\mathcal{O}$,
   then  $\mathcal{O}$ satisfies $\{e_lf_l\}\{e_pe_q\}$ if and only if it does not satisfy $\{e_lf_l\}\{f_sf_r\}$.
    By the symmetry there are four cases to check (see Figure~\ref{fig:thetaTree5}).
   Using the language of the formal logic the previous fact about $\mathcal{O}$
   is expressed by the following formula.

$$  \{e_pf_r\}\{e_qf_s\} \Rightarrow  (\{e_lf_l\}\{e_pe_q\} \Leftrightarrow \{e_lf_l\}\{f_sf_r\})  \ \bigwedge \  \neg \{e_pf_r\}\{e_qf_s\} \Rightarrow  (\{e_lf_l\}\{e_pe_q\} \Leftrightarrow \neg \{e_lf_l\}\{f_sf_r\}) $$
   
   \begin{figure}[htp]
\centering
\includegraphics[scale=0.7]{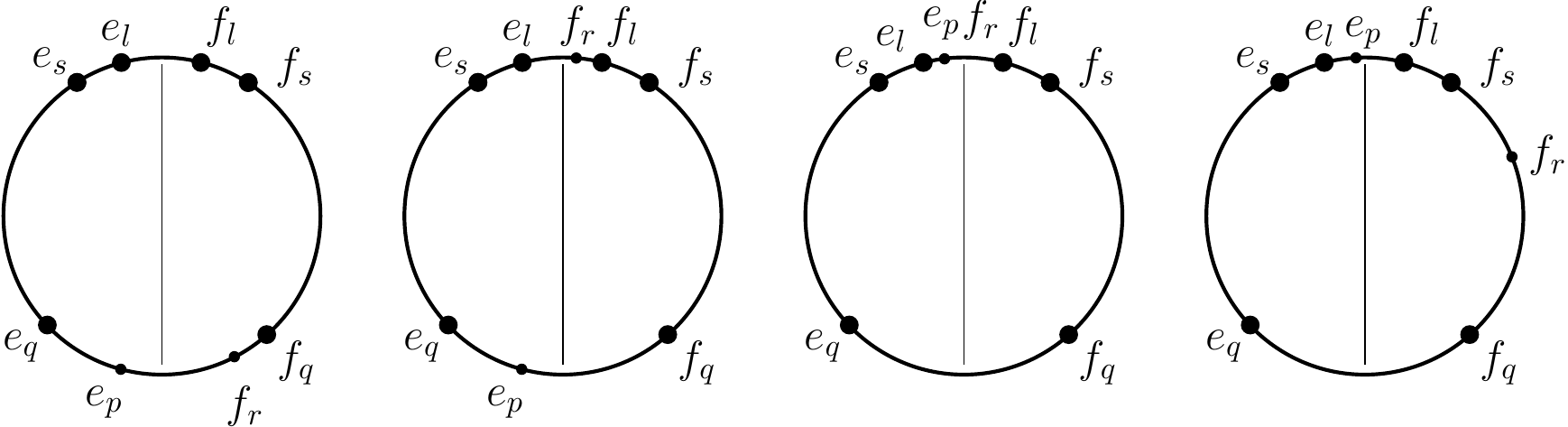}
\caption{The orderings of the leaves in $T_0$ corresponding to the true assignment of the propositional formula. Since only $e_p,e_q,f_r,f_s,e_l$ and $e_l$ appear in the formula.
We can fix $e_q,f_s,e_l$ and $f_l$ and use top-bottom symmetry.}
\label{fig:thetaTree5} 
\end{figure}

By the formula, it easily follows that the orientation  around $u$ and $v$, respectively, of the edges $e_l,e_p,e_q$ and $f_l,f_r,f_s$ is opposite of each other and we are done.
\end{proof}

\section{Beyond theta graphs and trees}
 
 One might wonder if our algorithm and/or Hanani--Tutte variant for theta graphs can be extended to 
arbitrary planar graphs.
 
 It is tempting to consider the following  definition of $(G',\gamma')$ corresponding 
 to an instance $(G,\gamma)$.
 Let us suppress every vertex of degree two in $G$. 
Let $G''$ denote the  resulting multi-graph. Each edge $e$ in $G''$ has a corresponding
path $P:=P(e)$ in $G$ (which is maybe just equal to $e$).
Let us assign a non-negative weights $w(e)$ to every edge $e$ of $G''$ equal to $\max(P)-\min(P)$.
Let $G'''$ denote the minimum weight spanning tree of $G''$.

The tree $G'$ is obtained as the union of the subgraph of $G$ corresponding to
$G'''$, and pairs of copies of paths $P(e)$, for each $e=uv$, where
one copy of each pair is attached to $u$ and the other one to $v$, but otherwise 
disjoint from $G'''$. The assignment of the vertices to the clusters  in $(G',\gamma')$ 
is inherited from $(G,\gamma)$.

Note that if $G$ is a theta graph the resulting $G'$ is almost equal to the one defined
in the previous subsections, except that earlier we shortened copies of $P(e)$'s for $e$ not in $G'''$
so that they do not contain one end vertex of $P(e)$. In the case of theta graphs this does
not make a difference, but for more general class of graphs our new definition of $G'$ might
be more convenient to work with. 

Now, we would like to use the construction of~\cite[Section 4.1]{BR14} enriched by
the constraints of $(G',\gamma')$, and other necessary contraptions if 
the graph is not two-connected.

If the graph is two-connected it seems plausible that
it is enough if we  generalize our construction of $T_1$ (taking care of trapped vertices),
and prove that the constraints of $(G',\gamma')$ together with other constraints account for 
all infeasible interleaving pairs of path.
To construct $T_1$ is not difficult, since we just repeat our construction of the matrix corresponding to $T_1$ for each consistency tree corresponding to a pair of vertices participating in the two-cut,
and combine the resulting matrices.

Then it seems that the only challenging part is to adapt the last paragraph in the proof of Theorem~\ref{thm:theta_alg}, which does not seem to be beyond reach. To this end it is likely that  a more efficient version of the characterization in Theorem~\ref{thm:characterization} is also needed, where by ``more efficient'' we mean a version
that restricts the set of interleaving pairs of paths considered. This should be possible,
since in our proofs for trees and theta graphs a considerably  restricted subset
of interleaving feasible pairs implied feasibility for the rest.

When the graph $G$ is not guaranteed to be two-connected, generalizing $T_1$ does not
seem to be a way to go. However, the general strategy of combining
simultaneous PC-orderings with our characterization could work, if the constraints preventing an occurrence of trapped vertices and infeasible interleaving pairs  are treated simultaneously using a more efficient version of 
 Theorem~\ref{thm:characterization}.

\section{Open problems}

\label{sec:open}

We also wonder if  the running time in Theorem~\ref{thm:AlgTreeXBounded} and~\ref{thm:theta} can be improved,
and if our main result, Theorem~\ref{thm:main}, can be extended
to radial~\cite{BBF04} or cyclic~\cite{ALBFMI15+,F16+} setting,
to higher genus surfaces and what are  its higher dimensional analogues.

\paragraph{Acknowledgment.}

We would like to express our sincere gratitude to the organizers and
participants of the 11th GWOP workshop, where we could discuss the research
problems treated in the present paper. In particular, we especially benefited from the discussions with Bettina Speckmann, Edgardo Rold\'{a}n-Pensado and Sebastian Stich.

Furthermore, we would like to thank J\'{a}n Kyn\v{c}l for useful discussions at the initial stage of this work and many useful comments,
 G\'abor Tardos for comments that helped to improve the presentation of the results, and to Arkadij Skopenkov for telling us about~\cite{S03}. An adaptation of the ideas from this work and also from~\cite{M97}, resulted in a tremendous simplification 
 of the proof of our main result, Theorem~\ref{thm:characterization}.

\bibliographystyle{plain}

\bibliography{bib}

\end{document}